\documentclass[12pt]{article}

\usepackage[sort,longnamesfirst]{natbib}

\usepackage{amsbsy,amsmath,amsthm,amssymb,booktabs,verbatim,color}

\usepackage{graphicx,subfigure,placeins}

\usepackage{geometry} 
\usepackage[dvips]{changebar}

\newtheorem{theorem}{Theorem}

\newtheorem{lemma}{Lemma}

\theoremstyle{remark}
\newtheorem{remark}{Remark}

\begin{document}

\title{Fully Bayesian Penalized Regression with a Generalized Bridge Prior} 

\author{Ding Xiang \\ School of Statistics\\ University of
  Minnesota\\ \texttt{xiang045@umn.edu}  \and Galin L. Jones \\ School of
  Statistics\\ University of Minnesota\\ \texttt{galin@umn.edu} }

\date{\today}

\maketitle
                               
\begin{abstract} 
We consider penalized regression models under a unified framework where the
particular method is determined by the form of the penalty term. We propose
a fully Bayesian approach that incorporates both sparse and dense settings and
show how to use a type of model averaging approach to eliminate the nuisance
penalty parameters and perform inference through the marginal posterior
distribution of the regression coefficients.  We establish tail robustness of
the resulting estimator as well as conditional and marginal posterior
consistency. We develop an efficient component-wise Markov chain Monte Carlo
algorithm for sampling. Numerical results show that the method tends to select
the optimal penalty and performs well in both variable selection and
prediction and is comparable to, and often better than alternative methods.
Both simulated and real data examples are provided.  
\end{abstract}

\newpage

\section{Introduction}
Penalized regression methods such as the lasso \citep{tibs:1996}, ridge
regression \citep{hoer:kenn:1970}, and bridge regression
\citep{fran:frie:1993, fu:1998} have become popular alternatives to ordinary
least squares (OLS).  All of these methods can be viewed in a common
framework. If $Y$ is a centered $n$-vector of responses, $X$ is a standardized
$n\times p$ matrix, and $\beta$ is a $p$-vector, then estimates are obtained
by solving
\[
\arg\min_{\beta}\left\{(Y-X \beta)^T(Y-X\beta)+\lambda||\beta||_\alpha \right\}\,,
\label{eq:penalty1}
\]
where $||\beta||_\alpha = \sum_{i=1}^p|\beta_i|^\alpha$, $\lambda\geq 0$ and
$\alpha \ge 0$. When $\lambda=0$ the OLS estimator is recovered, while if
$\lambda>0$, then $\alpha=1$ corresponds to the lasso, $1< \alpha <2$
corresponds to bridge regression, and $\alpha=2$ corresponds to ridge
regression. Now $\alpha < 1$ is useful in sparse settings
\citep{zheng:etal:2015} but non-convexity has limited its application.   

While it has become routine to choose $\lambda$  using cross validation on a
grid of possible values, the choice of $\alpha$ is complicated by the fact
that each method performs best in different regimes defined by the nature of
the unknown true parameter $\beta^0$ and whether the goal is variable
selection, estimation, or prediction \citep{fu:1998, hastie:etal:ESL:2009,
hastie:etal:SLS:2015, tibs:1996, wang:etal:2019, zou:hast:2005}. Moreover, the
dominant view of estimating $\alpha$ is apparently that ``...it is not worth
the effort...'' \citep[][p. 72]{hastie:etal:ESL:2009}.  Thus the default
approach in applications has been to preselect  $\alpha=1$ or $\alpha=2$ or
perhaps choose between them using cross validation. 

Bayesian approaches to penalized regression methods also have received much
recent attention.  \cite{tibs:1996} characterized the lasso estimates as a
posterior quantity, however, the first explicit Bayesian approach to lasso
regression is introduced by \cite{park:case:2008} followed by \cite{hans:2009}
and \cite{kyun:etal:2010}.  \cite{fu:1998} and \cite{pols:etal:2014} studied
Bayesian bridge regression while \cite{case:1980}, \cite{fran:frie:1993}, and
\cite{grif:brow:2013} considered Bayesian ridge regression. Of course, 
Bayesian approaches also require a choice of $\lambda$ and $\alpha$. Some
Bayesian approaches have incorporated a prior for $\lambda$, some have used
empirical Bayes approaches to estimate it, and some have conditioned on it
\citep{case:1980, hans:2009, khar:hobe:2013, kyun:etal:2010, park:case:2008,
roy:chak:2017}. On the other hand, there has been little investigation of how
to deal with $\alpha$. \citet{pols:etal:2014} considered priors for $\alpha
\in (0,1)$, but other Bayesian methods condition on the choice of $\alpha$
through preselection.

There have been a number of other Bayesian approaches to linear regression for
sparse signal detection. These have typically centered around spike-and-slab
priors \citep{geor:mccu:1993, mitc:beau:1989, nari:he:2014, rock:george:2016}
and continuous shrinkage priors \citep{carv:etal:2010, grif:brow:2017,
pols:etal:2010, fabr:etal:2010, sala:etal:2012, grif:hoff:2017}.  

We propose a fully Bayesian approach to penalized regression that incorporates
both sparse and dense settings and show how to use a type of model averaging
approach to eliminate the nuisance penalty parameters $(\lambda, \alpha)$ and
perform inference through the marginal posterior distribution of the
regression coefficients.  Although we use a version of spike-and-slab priors
we will see that our approach has more in common with local-global priors.  In
particular, we show that our prior has a local-global interpretation and leads
to the same sort of tail-robustness properties enjoyed by the horseshoe prior
\citep{carv:etal:2010}.  We also consider the setting where dimension grows
with sample size and establish both conditional and marginal strong posterior
consistency.

We explore the properties of the proposed model via simulation and compare it
to a number of alternatives such as Bayesian and frequentist versions of lasso
and ridge regression as well as the horseshoe estimator \citep{carv:etal:2010}
and spike-and-slab lasso regression \citep{rock:george:2016}.  We will
demonstrate that our approach results in estimation and prediction that is
comparable to, and often better than, existing methods.  Moreover, while our
approach performs well in sparse settings, our simulation results also show
that it performs well in dense settings.  

Our starting point is the standard Bayesian formulation of penalized
regression models which assumes
\begin{equation*}
Y|X,\beta,\gamma \sim \text{N}(X\beta, \gamma^{-1}I_n),
\end{equation*}
with $I_n$ an $n \times n$ identity matrix, along with priors $\nu(\gamma)\propto \gamma^{-1} $ and
\begin{align}
\nu(\beta \vert \gamma, \lambda,\alpha) = \left(
  \dfrac{\alpha(\gamma\lambda)^{1/\alpha}}{2^{1/\alpha+1}\Gamma(1/\alpha)}
\right)^p \exp\left\{-\frac{\gamma\lambda}{2}\|\beta\|_{\alpha}\right\}\, .
\label{eq:beta_intro}
\end{align}
Notice that if $Y=y$ is observed and $(\lambda, \alpha)$ is fixed, this yields a marginal posterior density 
\begin{equation}
\label{eq:cond post}
q(\beta|y) \propto \left[(y-X\beta)^T(y-X\beta)+\lambda \|\beta\|_\alpha
\right]^{-[n/2 + p/\alpha]}  
\end{equation}
from which one can easily observe that the estimator obtained in
\eqref{eq:penalty1} amounts to the posterior mode and is thus suboptimal under
squared error loss for which the Bayes (optimal) estimator is the posterior
mean; see \cite{hans:2009} for a clear discussion on this point in the context
of the Bayesian lasso and \cite{berg:1985} for more general settings.

We propose a fully Bayesian hierarchical model using a more general version of
the prior in \eqref{eq:beta_intro} and incorporating a prior for $(\lambda,
\alpha) \in [0,\infty)^p \times [k_1,k_2]$, where $0<k_1\leq 1,\ 2 \leq k_2 $,
which yields a posterior density $q(\beta, \gamma, \lambda, \alpha |y)$.
Allowing $k_1$ to be less than 1 will encourage sparsity when appropriate,
while allowing $k_2$ to be larger than 2 will yield improved performance in
dense settings.  This fully Bayesian approach encourages inference to proceed
naturally using a type of model-averaging.  If estimation of the true value of
$\beta$ is of interest, then the marginal density $q(\beta |y)$ can be used to
produce an estimate along with posterior credible intervals. If prediction of
a future value $\tilde{Y}$ is desired we calculate the posterior mean of the
posterior predictive density while prediction intervals based on the posterior
predictive density are conceptually straightforward.  

We can also use the hierarchical model to perform inference about $(\lambda,
\alpha)$ based on the appropriate marginal density.  Consider estimation of
$\alpha$. In Section~\ref{sec:examples} we conduct a simulation study where
four scenarios are identified such that in scenario I  and IV the lasso should
be preferred, while in scenario II ridge and lasso should be comparable, and
in scenario III ridge should be preferred.  The estimated marginal posterior
density for $q(\alpha|y)$ for a single simulated data set from each scenario
is displayed in Figure~\ref{fig:alpha_posterior}.  We see that the posterior
density tends to have most of its mass near the values of $\alpha$
corresponding to the optimal penalization method.   These results were typical
in our simulations.

\begin{figure}[ht!]
 \centering
  \includegraphics[width=0.7\textwidth]{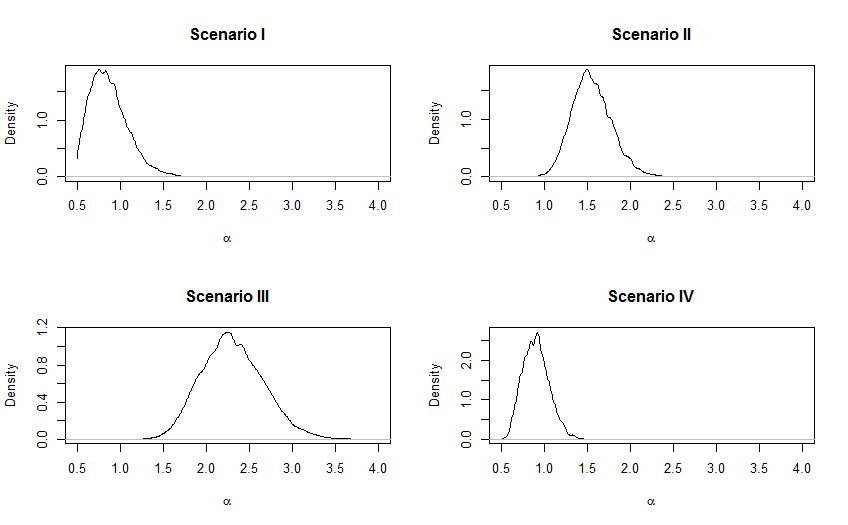}
\caption{Estimated marginal posterior density of $\alpha$ in four scenarios with a uniform prior on $[0.5, 4]$.}
 \label{fig:alpha_posterior}
\end{figure}

The posterior for the proposed hierarchical model is analytically intractable in the sense that it is difficult to calculate the required posterior quantities.  Thus we develop an efficient component-wise Markov chain Monte Carlo (MCMC) algorithm \citep{john:jone:neat:2013} to sample from the posterior.  We also consider Monte Carlo approaches to estimating posterior credible intervals and interval estimates based on the posterior predictive distribution. 

%%%%%%%%%%%%%%%
%\begin{comment}
The rest of the paper is organized as follows. In Section~\ref{sec:hm} we introduce the hierarchical model. Then we turn our attention to some theoretical properties of the model by establishing certain tail robustness properties in Section~\ref{sec:shrinkage} and then studying strong posterior consistency in Section~\ref{sec:consistency}.  Section~\ref{sec:est and pred} addresses estimation and prediction with a Markov chain Monte Carlo algorithm. Simulation experiments and a data example are presented in Sections~\ref{sec:examples} and~\ref{sec:diabetes}, respectively. Some final remarks are given in Section~\ref{sec:remarks}. All proofs are deferred to the appendix.
%\end{comment}
%%%%%%%%%%%%

\section{Hierarchical Model}
\label{sec:hm}

%%%%%%%%%%%%%%
\begin{comment}
%%%%%%%%%%%%%%
We will construct a hierarchical model which shrinks small effects while simultaneously introducing little bias for larger effects. Indeed the posterior consistency results of Section~\ref{sec:consistency} give sufficient conditions for this and we demonstrate it via simulation in Section~\ref{sec:examples}. 
%%%%%%%%%%%%%%
\end{comment}
%%%%%%%%%%%%%%

We continue to assume the response $Y$ follows a normal distribution 
\begin{equation}
Y|X,\beta,\gamma \sim \text{N}(X\beta, \gamma^{-1}I_n)\, .
\label{eq:likelihood}
\end{equation}
We also assume a proper conjugate prior $\gamma \sim \text{Gamma}(e_3 , \, f_3)$. Next, we assume
\begin{equation}
\nu(\beta \vert \gamma, \lambda,\alpha) = \left(\dfrac{\alpha(\gamma)^{1/\alpha}}{2^{1/\alpha+1}\Gamma(1/\alpha)}\right)^{p}
\left(\prod_{i=1}^p\lambda_i\right)^{1/\alpha}
\exp\left\{-\frac{\gamma}{2}\sum_{i=1}^p\lambda_i|\beta_i|^\alpha\right\}.  \label{eq:beta prior} 
\end{equation}
The only difference from \eqref{eq:beta_intro} is that for each $\beta_i$ we assign a parameter $\lambda_i \ge 0$, which allows for differing shrinkage in estimating each component.  Figure~\ref{fig:beta_prior} displays the density for some settings of $\alpha$ and $\lambda$. 
\begin{figure}[ht!]
 \centering\includegraphics[scale=.6]{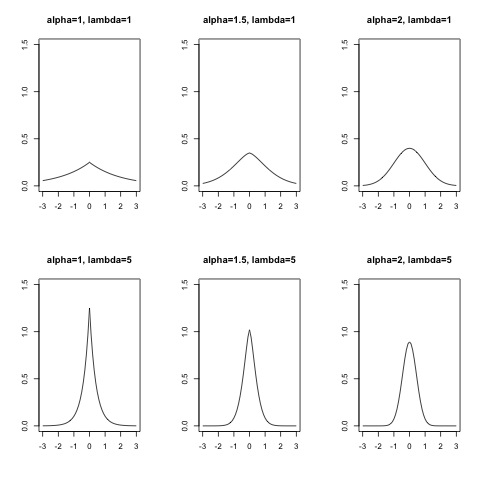}
 \caption{Prior density of $\beta_i$ for different values of $\alpha$ and $\lambda_i$ when $\gamma = 1$.}
 \label{fig:beta_prior}
\end{figure}
 Routine calculation shows that $E(\beta_i |\gamma, \lambda_i, \alpha)=0$ and 
\begin{align*} 
\text{Var}(\beta_i \vert \gamma, \lambda_i,\alpha)= \frac{\Gamma(3/\alpha)}{\Gamma(1/\alpha)} (\gamma \lambda_i)^{-2/\alpha}4^{1/\alpha} \, .   
\end{align*} 
Hence the variance is a decreasing function of $\lambda_i$. If $\lambda_i$ is small, larger values of $\beta_i$ are likely but if $\lambda_i$ is large, smaller values of $\beta_i$ are likely.  This suggests a way to incorporate a spike-and-slab prior through the prior for $\lambda_i$.  Specifically, we assume
\begin{align}
 \label{eq:prior_lambda}
  \nu(\lambda_i|\kappa_i, e_1, f_1, e_2, f_2) & = (1-\kappa_i) \text{Gamma}(\lambda_i; e_1, f_1) + \kappa_i \text{Gamma}(\lambda_i; e_2, f_2)
\end{align}
and $\nu(\kappa_i) \sim \text{Bern}(1/2)$.  The hyperparameters are chosen so that one component of the mixture has a small mean and variance while the other can have a relatively large mean and variance.

Finally, we need to specify a prior for $\alpha$.  Notice that, unlike $\lambda_i$ which controls shrinkage for an individual $\beta_i$, the parameter  $\alpha$ is common to all of the $\beta_i$. If one wants to stay with the analogy with the frequentist methods in \eqref{eq:penalty1}, then it is natural to assume 
\begin{align}
\nu(\alpha|c_1,c_2, c_3)=& c_1\text{Beta}(\alpha-1, a_1, b_1) + c_2\text{Beta}(\alpha-1;a_2, b_2) + c_3\text{Beta}(\alpha-1;a_3,b_3),
\label{alpr_org}
\end{align}
where $\text{Beta}(\alpha-1, a, b)$ is a $\text{Beta}( a, b)$ shifted to have support on $[1,2]$ and each $c_j\in [0, 1]$ such that $\sum_{j=1}^3 c_j=1$.  The idea here is that each component represents the analyst's assessment of the relative importance of lasso, bridge, and ridge, but our empirical work indicated that different choices yield similar estimation and prediction. This motivated us to consider a uniform distribution  for $\alpha$ which we have found to work well,  especially since extending the range of $\alpha$ appears to be impactful. Therefore we assume
\begin{align}
\alpha \sim &\text{Unif}(k_1, k_2) ~~~~k_1\leq 1,~~ 2 \leq k_2 .
\label{alpr}
\end{align}
One expects that $\alpha < 1$ will encourage even more sparse results and recover best subset selection for small $\alpha$.  In our experience estimation and prediction performance are similar among different choices of $k_1$. However, allowing $k_2 > 2$ is especially helpful in dense settings with small effects.  Consider Figure~\ref{fig:beta_alpha_margin_1} and \ref{fig:beta_alpha_margin_3} which are contour plots of the joint posterior density for $(\beta, \alpha)$ when $\beta$ is a scalar. We simulated data $y$ under the assumption of the true $\beta^0 =0.5$. In Figure~\ref{fig:beta_alpha_margin_1}, we have the ordinary range of $\alpha \in [0.1, 2]$. The posterior distribution clearly concentrates near $(\alpha=0.5, \beta=0)$, which would lead us to estimate $\beta$ with 0. In Figure~\ref{fig:beta_alpha_margin_3}, we expand the range to have $\alpha \in [0.1, 8]$ while keeping $e$ and $f$ unchanged.  In this case we see that the posterior does not concentrate near $\beta=0$ and hence will allow us to more reasonably estimate small nonzero effects.

\begin{figure}[ht!]
 \centering\includegraphics[scale=.5]{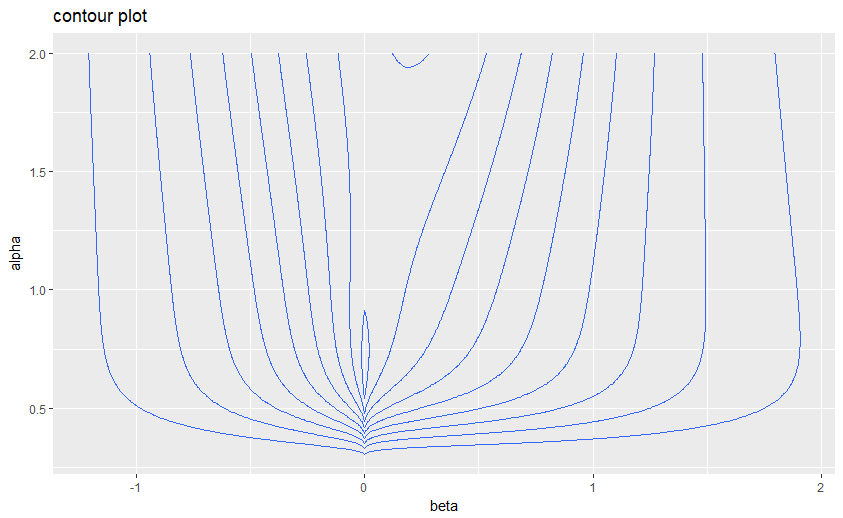}
 \caption{$k_1=0.1$, $k_2=2$.}
 \label{fig:beta_alpha_margin_1}
\end{figure}

\begin{figure}[ht!]
 \centering\includegraphics[scale=.5]{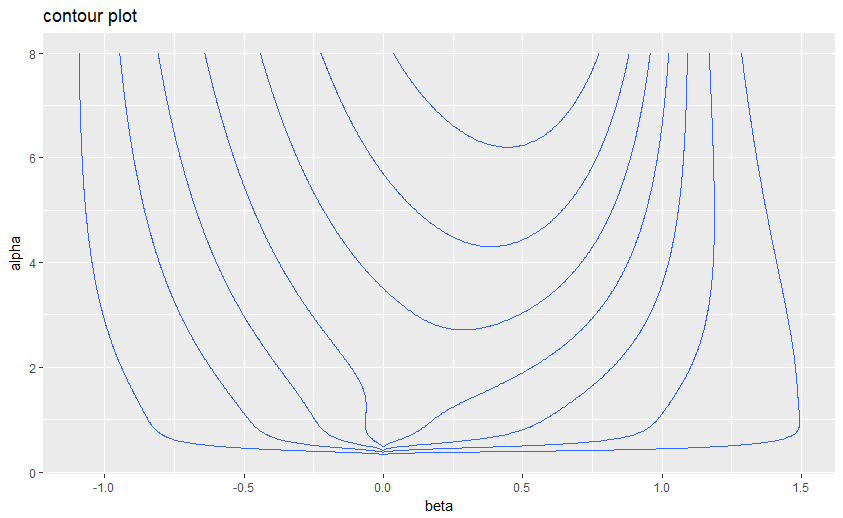}
 \caption{$k_1=0.1$, $k_2=8$.}
 \label{fig:beta_alpha_margin_3}
\end{figure}

\FloatBarrier
\section{Theory}
\label{sec:theory}
In this section we consider two theoretical properties of the posterior.  We begin by establishing a tail robustness property similar to that of the horseshoe prior and then we turn our attention to posterior consistency. 

\subsection{Tail Robustness}
\label{sec:shrinkage}
Consider the following  one-dimensional case version of the model above
\begin{align*}
  Y | \beta & \sim \text{N}(\beta, 1)\\
  \nu(\beta \vert \gamma, \lambda,\alpha) &= \dfrac{\alpha(\gamma)^{1/\alpha}}{2^{1/\alpha+1}\Gamma(1/\alpha)}
                                            \lambda^{1/\alpha}
                                            \exp\left\{-\frac{\gamma}{2}\lambda|\beta|^\alpha\right\}.\\
 \lambda & \sim (1-\kappa)\text{Gamma}(\lambda; e_1, f_1) + \kappa \text{Gamma}(\lambda; e_2, f_2)\\
 \alpha &\sim \text{Unif}(k_1, k_2) \\
 \kappa &\sim \text{Bern}(1/2) .
\end{align*}
Let $m(y)$ be the marginal density achieved by integrating over all the parameters.  A standard calculation shows that the marginal posterior mean of $\beta$ satisfies
$$E(\beta|y) = y + \frac{d}{dy}\log m(y)$$ 
and hence the following result shows that our priors satisfy a tail-robustness property.

\begin{theorem}
\label{thm:shrinkage}
There is some $C_h$ which depends on the hyperparameters such that 
$|y-E(\beta|y)|\leq C_{h}$ and 
$$\lim_{|y|\rightarrow \infty}\frac{d}{dy}\log m(y) = 0 \, .$$
\end{theorem}
\begin{proof}
See Appendix~\ref{app:proof:shrinkage}.
\end{proof}

\subsection{Posterior Consistency}
\label{sec:consistency}
We establish sufficient conditions for the posterior to concentrate near the true regression coefficients as the dimension grows with sample size.  We slightly modify our notation to make the dependence on the sample size explicit.  Let $\theta_n = \{ \lambda_n , \, \alpha_n, \, \kappa_n\}$ and let $h_n$ denote all of the hyperparameters.  Then the full posterior distribution is denoted $Q_n (\beta_n , \theta_n | y_n, h_n)$ since, in this section, we assume the precision $\gamma$ is known.  We will establish both consistency with respect to the marginal $Q_n (\beta_n | y_n , h_n)$ and consistency with respect to the conditional $Q_n (\beta_n | \theta_n, y_n, h_n)$.

We make the following assumptions throughout this section (i) $p_n = o(n)$, as $n \to \infty$; (ii) if $\Lambda_{n\min}$ and $\Lambda_{n\max}$ are the smallest and the largest singular values of $X_n$, respectively, then $0<\Lambda_{\min}<\liminf_{n \rightarrow \infty} \Lambda_{n\min}/\sqrt{n}\leq\limsup_{n\rightarrow\infty}\Lambda_{n\max}/\sqrt{n}<\Lambda_{\max}<\infty$; (iii) if $\beta_n^0$ is the true regression parameter, then $\sup_{j=1,\cdots,p_n}|\beta_{nj}^0| <\infty$; and  if $m_n$ denotes the number of nonzero elements in $\beta_n^0$, then $m_n=o\{n^{1-\rho}/(p_n\log^2 n)\}$, as $n \to \infty$, for $\rho\in (0,1)$.  Finally, let $F_{\beta^0_n}$ denote the distribution at \eqref{eq:likelihood} under the true regression parameter and for $\epsilon >0$ set
\[
B_{n,\epsilon} = \{ \beta_n \, : \, \|\beta_n - \beta^0_n\| > \epsilon\}\, .
\] 
We are now in position to state our  result on conditional consistency.

\begin{theorem}
\label{thm:conditional}
If, for each $j\in [1, p_n]$,  $\lambda_{nj}=(C\sqrt{p_n}n^{\rho/2}\log n)^{\alpha_n}$ for finite $C>0$, then for any $\epsilon >0$, as $n\rightarrow\infty$,
\[
Q_n (B_{n, \epsilon} | \theta_n, y_n, h_n) \to 0 ~~~~F_{\beta^0_n}-\text{almost surely} \; .
\]
\end{theorem}
\begin{proof}
See Appendix~\ref{app:proofs}.
\end{proof}
Next we address marginal posterior consistency. 
\begin{theorem}
\label{thm:marginal}
If, for each $j\in [1, p_n]$, each element $\lambda_{nj}=(C_n \sqrt{p_n}n^{\rho/2}\log n)^{\alpha_n}$ for $C_n>0$ and $C_n^2 =o(n)$, then for any $\epsilon>0$, as $n \to \infty$,
\[
Q_n (B_{n, \epsilon} |y_n, h_n) \to 0 ~~~~F_{\beta^0_n}-\text{almost surely}\; .
\]
\end{theorem}
\begin{proof}
See Appendix~\ref{app:proofs}.
\end{proof}

\section{Estimation and Prediction}
\label{sec:est and pred}
The hierarchical model gives rise to a posterior density characterized by
\begin{equation}
\label{eq:post}
q(\beta, \gamma, \lambda, \alpha, \kappa |y) \propto f(y|\beta) \nu(\beta|\alpha, \lambda) \nu(\lambda|\kappa) \nu(\gamma) \nu(\alpha) \nu(\kappa) 
\end{equation}
and which yields marginal density $q(\beta | y)$.  Under squared error loss, the Bayes (optimal) estimator of the regression coefficients is $\hat{\beta}=E[\beta|y]$.  Interval estimates can be constructed from quantiles of the posterior marginal distribution of $\beta |y$. Similarly, we can estimate and make inference about the other parameters through the appropriate marginal distributions.

Under squared error loss prediction of a future observation $\tilde{Y}$ is based on the mean of the posterior predictive distribution
\begin{equation}
\label{eq:post_pred}
E[\tilde{Y}|y] = \int \tilde{y} q(\tilde{y}|y) d\tilde{y} = \int \tilde{y} f(\tilde{y}|\beta, \gamma) q(\beta, \gamma, \lambda, \alpha, \kappa|y) d\beta\, d\gamma\, d\lambda\, d\alpha\, d\kappa\, d\tilde{y} \; .
\end{equation}
A routine calculation shows that if $\tilde{X}$ corresponds to a new observation, then $E[\tilde{Y}|y] =\tilde{X} E[\beta|y] = \tilde{X} \hat{\beta}$.  Interval estimates can be constructed from quantiles of the posterior predictive distribution.

Unfortunately, calculation of $\hat{\beta}$ and quantiles of  posterior marginals or the posterior predictive distribution $q(\beta|y)$ and $q(\tilde{y}|y)$ is analytically intractable so we will have to resort to Monte Carlo methods, which are considered in the sequel.

\subsection{Markov Chain Monte Carlo}
\label{sec:mcmc}
We develop a deterministic scan component-wise MCMC algorithm with invariant density $q(\beta, \gamma, \lambda, \alpha, \kappa |y)$ which consists of a mixture of Gibbs updates and Metropolis-Hastings updates.  To begin we require the posterior full conditionals.  Let $\beta_{-i}$ be all of the entries of $\beta$ except $\beta_{i}$.  Then
\begin{align}
q(\beta_i \vert  \beta_{-i}, \alpha, \gamma, \lambda_i) \propto \exp\left\{-\frac{\gamma\lambda_i}{2}\vert\beta_i\vert^{\alpha}\right\} \exp\left\{-\frac{\gamma}{2}( Y- X \beta)^T( Y- X\beta)\right\}\, ,
\label{eq:beta2}
\end{align}
\begin{align}
q(\alpha \vert   \beta, \gamma, \lambda_i) \propto \left(\dfrac{\alpha \gamma^{1/\alpha}}{2^{1/\alpha} \Gamma(1/\alpha)}\right)^{p} \left(\prod_{i=1}^p\lambda_i\right)^{1/\alpha} \exp\left\{-\frac{\gamma}{2}\sum_{i=1}^p\lambda_i|\beta_i|^\alpha\right\} \nu(\alpha) \, ,
\label{eq:alpha2}
\end{align}
\begin{align*}
\gamma \vert  \beta, \alpha, \lambda_i  \sim\text{Gamma} \left(e_3+ \frac{n}{2}+ \frac{p}{\alpha}, f_3 + \frac{1}{2} \left[( Y- X\beta)^T( Y- X\beta)+ \sum_{i=1}^p\lambda_i|\beta_i|^\alpha \right] \right) \, ,
\end{align*}
\begin{align*}
\lambda_i \vert   \beta, \alpha, \gamma,  \kappa_i \sim (1-\kappa_i)\text{Gamma} \left(e_1+\frac{1}{\alpha}, f_1+\dfrac{\gamma}{2}|\beta_i|^\alpha \right)+\kappa_i\text{Gamma} \left(e_2+\frac{1}{\alpha}, f_2+\dfrac{\gamma}{2}|\beta_i|^\alpha \right) \, ,
\end{align*}
and
\begin{align*}
\kappa_i|\lambda_i \sim \text{Bern} \left( \frac{\omega_2}{\omega_1+\omega_2} \right) 
\end{align*}
where $\omega_1=\text{Gamma}(\lambda_i; e_1,f_1)$ and $\omega_2=\text{Gamma}(\lambda_i; e_2,f_2)$ are Gamma densities evaluated at $\lambda_i$.  We see that we can use Gibbs updates for $\gamma$, the $\lambda_i$ and the $\kappa_i$.  However, for the $\beta_i$ and $\alpha$ we will need Metropolis-Hastings updates, which are now described.
 
Consider updating $\beta_i$.  If $\beta_i^{(t)}$ is the current value at the $t$th iteration, then we will use a random walk Metropolis-Hastings update with proposal distribution $N(\beta_i^{(t)},v_b)$, where $v_b$ is chosen by the user, and invariant density given by \eqref{eq:beta2}. 

The MH update for $\alpha$ is straightforward. We use an independence Metropolis-Hastings sampler with  invariant density given by \eqref{eq:alpha2}.

Cycling through these updates for $M$ steps in the usual fashion yields an MCMC sample 
\[
\left\{\beta^{(t)}, \gamma^{(t)}, \alpha^{(t)}, \lambda^{(t)}, \kappa^{(t)} \right\}_{t=1}^{M} \, .
\]
Estimation is straightforward since the sample mean is strongly consistent for $E[\beta|y]$, that is, as $M \to \infty$,
\[
\frac{1}{M} \sum_{t=1}^{M} \beta^{(t)} \to E[\beta|y] ~~~\text{with probability} ~1
\]
and a sample quantile of the $\{\beta^{(t)}\}$ is strongly consistent for the corresponding quantile of the marginal distribution \citep{doss:etal:2014}.  

Prediction intervals for a new observation require a further Monte Carlo step.  Consider the posterior predictive density 
\[
q(\tilde{y}|y) = \int f(\tilde{y}|\beta, \gamma) q(\beta, \gamma, \lambda, \alpha, \kappa|y) d\beta\, d\gamma\, d\lambda\, d\alpha\, d\kappa =  \int f(\tilde{y}|\beta, \gamma) q(\beta, \gamma |y) d\beta\, d\gamma
\]
so that given the MCMC sample we can sample from $q(\tilde{y}|y)$ by drawing 
$\tilde{Y}^{(t)} \sim f(\cdot|\beta^{(t)}, \gamma^{(t)})$ for $t=1,\ldots, M$.  The sample quantiles of $\{\tilde{y}^{(t)}\}_{t=1}^{M}$ are then strongly consistent for the corresponding quantiles of the posterior predictive distribution.  

\begin{remark}
If interest lies in extreme quantiles, then importance sampling is preferred \citep[see e.g.][]{robe:case:2013}.  However, for standard settings such as .05 or .95 quantiles, then the approach suggested here will be much faster.  In fact, compared to the above approach, our implementation of importance sampling with a Cauchy instrumental distribution was more than 550 times slower in our examples from Section~\ref{sec:examples} and hence we do not pursue it further here.
\end{remark}

\section{Simulation Experiments}
\label{sec:examples}
%We investigate our method's capability of addressing $\alpha$, estimating coefficients, and predicting in this section. In general, 
\subsection{Simulation Scenarios}
\label{sec:scenarios}
We consider four scenarios: (I) a small number of large effects; (II) a small to moderate number of moderate-sized effects; (III) a large number of small effects; and (IV) a sparse setting with $p>n$.  For each scenario we independently repeat the following procedure 500 times.  We generate 1000 observations from a model and split them into a training set of size $n_{\text{train}}$ and a test set of size $n_{\text{test}}$.  We then fit the hierarchical model from Section~\ref{sec:hm} on the training data using the MCMC algorithm and estimation procedure from Section~\ref{sec:est and pred}.  The hyperparameters were taken to be $k_1=0.5$, $k_2=4$, $e_1 = f_1 = 1$, $e_2 = 40$, $f_2 = 0.5$, and $e_3=f_3=0.001$.  The MCMC algorithm is run for 1e5 iterations, a value which was chosen based on obtaining enough effective samples according to the procedure developed by \cite{vats:etal:2017}.  The MCMC procedure is not computationally onerous since in our most challenging simulation experiment it took only a few seconds to complete for a single data set.

 In each scenario, we generate data from the following linear model:
\begin{align*}
Y = X\beta^0 + 2\varepsilon~~~~\varepsilon\sim \text{N}(0,I)\, .
\end{align*}
We include an intercept so that the first column of the design matrix $X$ is a column of ones.  The remaining columns are generated from a multivariate normal distribution $N_{p-1}(0, \Sigma)$, where the diagonal entries of $\Sigma$ equal 1 and the off-diagonals are $0.5^{|i-j|}$ for all $i,\ j \geq 2$.
Notice that $\beta^0$ is $(p+1)\times1$.

\noindent\textit{Scenario I}.  We set $p=20$ and, in each replication, randomly choose 18 of the 20 coefficients to be 0, while the remaining two are independently sampled from a $\text{N}(15,3^2)$.  Here $n_{\text{train}}=100$ and  $n_{\text{test}}=900$.

\noindent\textit{Scenario II}.  We set $p=20$ and, in each replication, randomly choose 10 of the 20 coefficients to be 0, while the remainder are independently sampled from a $\text{N}(5, 1)$.  Here $n_{\text{train}}=100$ and  $n_{\text{test}}=900$.

\noindent\textit{Scenario III}.  We set $p=20$ and, in each replication, all the coefficients are independently sampled from a $\text{N}(2,0.001^2)$.  Here $n_{\text{train}}=100$ and  $n_{\text{test}}=900$.

\noindent\textit{Scenario IV}.  We set $p=150$ and, in each replication, randomly choose 142 of the 150 coefficients to be 0, while the remaining  are independently sampled from a $\text{N}(15,3^2)$.  Here $n_{\text{train}}=50$ and  $n_{\text{test}}=950$.

\begin{remark} 
While we assume Gaussian errors in our simulation experiments, we also investigated the situation where this assumption is violated.  In particular, we considered the case where $\varepsilon$ follows a Student's $t$-distribution with 5 degrees of freedom.  In this setting our method continued to provide reasonable estimation and prediction. In fact, the results were similar enough that we do not present them here in the interest of a concise presentation.  
\end{remark}

\subsection{Posterior of $\alpha |y$}
\label{sec:alpha}

Recall Figure~\ref{fig:alpha_posterior} which displays the estimated posterior density of $q(\alpha|y)$ for a single data set in each of the four scenarios. The results coincide nicely with previous conclusions \citep{tibs:1996, hastie:etal:ESL:2009}.  In scenario I and IV where the lasso is preferred, more mass is close to 1. In scenario III ridge regression should dominate the lasso and $\alpha$ is concentrated in the region between 1.8 and 2.  In scenario II ridge and lasso are often comparable with a small advantage for lasso.  In the data set displayed here the estimated density favors larger values of $\alpha$, but we will see that the performance Bayesian methods are comparable to the optimal frequentist method. Overall, the proposed approach has the ability to provide a posterior density curve of $\alpha|y$ which puts most if its mass near the optimal values of $\alpha$. 

\subsection{Estimation}
\label{sec:est}
We compare the generalized bridge prior model with a Bayesian lasso (i.e. the hierarchical model of Section~\ref{sec:hm} with $\alpha=1$) and a Bayesian ridge regression (i.e. the hierarchical model of Section~\ref{sec:hm} with $\alpha=2$).  We also compare to the frequentist lasso and ridge regression with the tuning parameters chosen by 10-fold cross validation. We also compare the proposed procedure with the spike-and-slab lasso \citep{rock:george:2016} and horseshoe prior \citep{carv:etal:2010} as benchmarks.   

 Table~\ref{tab:estimation_L2} reports the average $L_2$ distance between estimated coefficients and the truth. These results suggest that the hierarchical model dominates the others when large effects exist (scenarios I and IV), especially in scenario IV where $n<p$. Our approach dominates ridge and is comparable to the others in scenario II.  In the scenario III, our model is superior to the other four models while ridge regression dominates. Largely this appears to be due to the hierarchical model more aggressively shrinking small effects to zero. 

\begin{table}[htbp]\small
   \centering

      \caption{Average $||\hat{\beta}-\beta||_2$ with standard errors.}
   \begin{tabular}{lllllllll}
      \toprule
   &  \multicolumn{8}{c}{Scenario} \\
 \cmidrule(lr){2-9}
Method & \multicolumn{2}{c}{I} &\multicolumn{2}{c}{II}&\multicolumn{2}{c}{III}&\multicolumn{2}{c}{IV} \\
       \cmidrule(lr){2-3}\cmidrule(lr){4-5}\cmidrule(lr){6-7}\cmidrule(lr){8-9}
      B.P. & 0.477 & 0.009 & 2.151 & 0.027 & 2.801 & 0.025  & 1.369 & 0.026\\ 
B.L. & 0.474 & 0.009  & 2.035 & 0.027 & 3.131 & 0.027  & 2.357 & 0.223 \\  
B.R. & 0.590 & 0.009 & 2.253 & 0.029 & 2.870 & 0.025 & 1.933 & 0.086 \\
Lasso & 0.999 & 0.016 & 2.315 & 0.028 & 3.255 & 0.029 & 6.406 & 0.108 \\ 
Ridge & 12.333 & 0.041 & 6.560 & 0.021 & 0.780 & 0.006 & 45.252 & 0.024\\ 
SSLasso & 0.943 & 0.036 & 2.144 & 0.032 & 3.201 & 0.032 & 3.560 & 0.062 \\ 
Horseshoe & 0.597 & 0.012 & 2.051 & 0.026 & 4.620 & 0.032 & 1.820 & 0.030  \\     
  
  \bottomrule
   \end{tabular}
   \label{tab:estimation_L2}
\end{table}

%Table \ref{tab:stnr} gives an estimate of the signal-to-noise ratio; specifically the ratio of the sum of squares due to regression (SSReg) divided by the sum of squared residuals (SSR). One can observe that the proposed method is at least as good as the alternatives presented here. 
%
%\begin{table}[htbp]
%   \centering
%      \caption{Estimated signal-to-noise ratio (SSReg/SSR) with standard errors.}
%   \begin{tabular}{lllllllllll}
%      \toprule
% Scenario  &  \multicolumn{10}{c}{Median SSReg/SSR} \\
% \cmidrule(lr){2-11}
% & \multicolumn{2}{c}{B.P.} &\multicolumn{2}{c}{B.L.}&\multicolumn{2}{c}{B.R.}&\multicolumn{2}{c}{Lasso} &\multicolumn{2}{c}{Ridge}\\
%       \cmidrule(lr){2-3}\cmidrule(lr){4-5}\cmidrule(lr){6-7}\cmidrule(lr){8-9}\cmidrule(lr){10-11}
%      I & 140.4 & 1.1 & 139.8 & 1.2 & 141.6 & 1.1 & 140.7 & 1.0 & 51.9 & 0.2 \\ 
%II & 179.1 & 1.4 & 180.7 & 1.3 & 180.7 & 1.2 & 175.6 & 1.4 & 131.4 & 0.8 \\  
%III & 42.6 & 0.4 & 42.9 & 0.4 & 42.8 & 0.4 & 45.0 & 0.5 & 40.2 & 0.4 \\
%IV & 1076.5 & 10.9 & 1048.7 & 9.6 & 1269.4 & 14.7 & 732.8 & 4.3 & 2.5 & 0.1 \\
%  \bottomrule
%   \end{tabular}
%   \label{tab:stnr}
%\end{table}

\subsubsection{Estimation of Large Effects}
 We consider two additional scenarios to study what happens in the presence of large effects.  All of the settings remain the same as above, except as noted below.

\noindent\textit{Scenario V}. Set $p=40$, $n_{\text{train}}=200$ and $n_{\text{test}}=400$. Values of regressors other than the intercept are also drawn from a multivariate distribution $N_{p-1}(0, \Sigma)$, where the diagonal entries of $\Sigma$ equals 1 and all off-diagonals are $0.5$. The simulation true vector of coefficients is follows, 
\begin{align*}
\beta^{T} = (\underbrace{0, \cdots, 0}_{10},\underbrace{2, \cdots, 2}_{10}, \underbrace{0, \cdots, 0}_{10},\underbrace{2, \cdots, 2}_{10}).
\end{align*}

\noindent\textit{Scenario VI}. We have the same settings as in scenario V except that
\begin{align*}
\beta^{T} = (\underbrace{0, \cdots, 0}_{10},\underbrace{150, \cdots, 150}_{10}, \underbrace{0, \cdots, 0}_{10},\underbrace{150, \cdots, 150}_{10}).
\end{align*}

Table~\ref{tab:estimation_L2_extra} reports the average $L_2$ distance between estimated coefficients and the truth in scenarios V and VI. Our Bayesian methods are all comparable and all dominate both the lasso and ridge regressions. Besides the two classic penalized regression, the Bayesian ridge model also failed to handle shrinkage on large coefficients efficiently. It is well known that lasso and ridge regression produce highly biased estimates in the presence of large effects, however, this does not appear to happen in the generalized bridge prior model.  Both spike-and-slab lasso and horseshoe prior models are slightly better than the generalized bridge prior model, at least on average. However, the generalized bridge prior model captures the true model 165 and 500 times out of the 500 replications, respectively, in these two scenarios, while the spike-and-slab lasso recovers the true model 132 and 127 times, respectively, and 2 and 52 times, respectively, by the horseshoe prior model.

\begin{table}[htbp]\small
   \centering
      \caption{Average $||\hat{\beta}-\beta||_2$ with standard errors.}
   \begin{tabular}{lllll}
      \toprule
   &  \multicolumn{4}{c}{Scenario} \\
 \cmidrule(lr){2-5}
Method & \multicolumn{2}{c}{V} &\multicolumn{2}{c}{VI}\\
       \cmidrule(lr){2-3}\cmidrule(lr){4-5}
      B.P. & 1.351 & 0.008 &  1.040 & 0.009\\ 
B.L. & 1.332 & 0.008  & 1.928 & 0.010 \\  
B.R. & 1.346 & 0.008 & 83.471 & 4.785\\
Lasso & 1.429 & 0.021 & 1560.272 & 95.456 \\ 
Ridge & 3.111 & 0.054 & 1426.277 & 288.195 \\ 
SSLasso & 1.026 & 0.009 & 1.036 & 0.009\\ 
Horseshoe & 0.964 & 0.007 & 0.953 & 0.007 \\  
  \bottomrule
   \end{tabular}
   \label{tab:estimation_L2_extra}
\end{table}
\begin{comment}
Table~\ref{tab:stnrle} gives an estimate of the of signal-to-noise ratio. Again, the proposed method is at least as good as the alternatives. 

\begin{table}[htbp]
   \centering
      \caption{Estimated signal-to-noise ratio (SSReg/SSR) with standard errors.}
   \begin{tabular}{lllllllllll}
      \toprule
 Scenario  &  \multicolumn{10}{c}{Median SSReg/SSR} \\
 \cmidrule(lr){2-11}
 & \multicolumn{2}{c}{B.P.} &\multicolumn{2}{c}{B.L.}&\multicolumn{2}{c}{B.R.}&\multicolumn{2}{c}{Lasso} &\multicolumn{2}{c}{Ridge}\\
       \cmidrule(lr){2-3}\cmidrule(lr){4-5}\cmidrule(lr){6-7}\cmidrule(lr){8-9}\cmidrule(lr){10-11}
V & 252.7 & 1.4 & 253.0 & 1.3 & 252.9 & 1.3 & 234.2 & 1.0 & 216.0 & 1.3 \\ 
VI & 1.43e6 & 7481.1 & 1.43e6 & 7322.7 & 1.43e6 & 7585.4 & 864.5 & 5.9 & 904.8 & 1.6 \\
  \bottomrule
   \end{tabular}
   \label{tab:stnrle}
\end{table}
\end{comment}

%%%%%%%%%%%%% PREDICTION %%%%%%%%%%%%%%%%%

\subsection{Prediction}
\label{sec:pred}
We turn our attention to prediction of a future observation.  The simulation results are based on the same simulated data as in Section~\ref{sec:est}.  Table~\ref{tab:prediction_mse} reports our simulation results.  The Bayesian approaches are comparable in all four scenarios, with the fully Bayesian approach being slightly better.  In scenario I the Bayesian methods dominate both lasso and ridge.  In scenario II the Bayesian methods are all comarable to lasso with ridge being substantially worse.  Ridge dominates in scenario III, but the other methods are comparable. In scenario IV, where $p > n$, both lasso and ridge regression are substantially worse.  

\begin{table}[htbp]\small
   \centering

      \caption{Median mean-squared error with standard errors.}
     \begin{tabular}{lllllllll}
      \toprule
   &  \multicolumn{8}{c}{Scenario} \\
 \cmidrule(lr){2-9}
Method & \multicolumn{2}{c}{I} &\multicolumn{2}{c}{II}&\multicolumn{2}{c}{III}&\multicolumn{2}{c}{IV} \\
       \cmidrule(lr){2-3}\cmidrule(lr){4-5}\cmidrule(lr){6-7}\cmidrule(lr){8-9}
      B.P. & 4.036 & 0.006 & 4.497 & 0.013  & 4.730 & 0.018  & 5.505 & 0.030\\ 
B.L. & 4.032 & 0.005  & 4.450 & 0.015 & 4.867 & 0.023  & 5.611 & 0.037 \\  
B.R. & 4.083 & 0.006  & 4.527 & 0.018 & 4.761 & 0.018  & 6.117 & 0.045  \\
Lasso  & 4.153 & 0.010 & 4.577 & 0.021 & 4.930 & 0.023  & 13.466 & 0.285 \\ 
Ridge & 16.385 & 0.111 & 7.413 & 0.043 & 4.329 & 0.011 & 1174.164 & 9.076\\ 
SSLasso & 4.118 & 0.019 & 4.510 & 0.021  & 4.919 & 0.023 & 7.977 & 0.122  \\ 
Horseshoe & 4.072 & 0.007  & 4.504 & 0.016 & 5.685 & 0.027 & 6.243 & 0.070  \\     
  
  \bottomrule
   \end{tabular}
 %  \caption*{$\dagger$ Standard errors of the medians are estimated via bootstrap with $B =1000$ resamplings on the 1000 mean-squared errors.}
   \label{tab:prediction_mse}
\end{table}

Figure~\ref{fig:pred_intvl} displays the posterior predictive densities and prediction intervals for a future observation.  In each graph, the dotted bell-shaped curve represents the true density function. The solid bell-shaped curve is the empirical posterior predictive distributions of $\tilde{Y}_1$.  The dotted line stands for the true $\tilde{Y}_1$s and the two dashed lines represent 2.5 and 97.5 percentiles respectively so that intervals in-between are the 95\%  intervals. Clearly the prediction interval contains the true value in each scenario.

\begin{figure}[htbp]
\caption{Posterior predictive intervals.  Scenario I is upper left, scenario II is upper right, scenario III  is lower left, and scenario IV is lower right.}
\label{fig:pred_intvl}
\begin{minipage}[b]{0.4\linewidth}
\centering
\includegraphics[width=\textwidth]{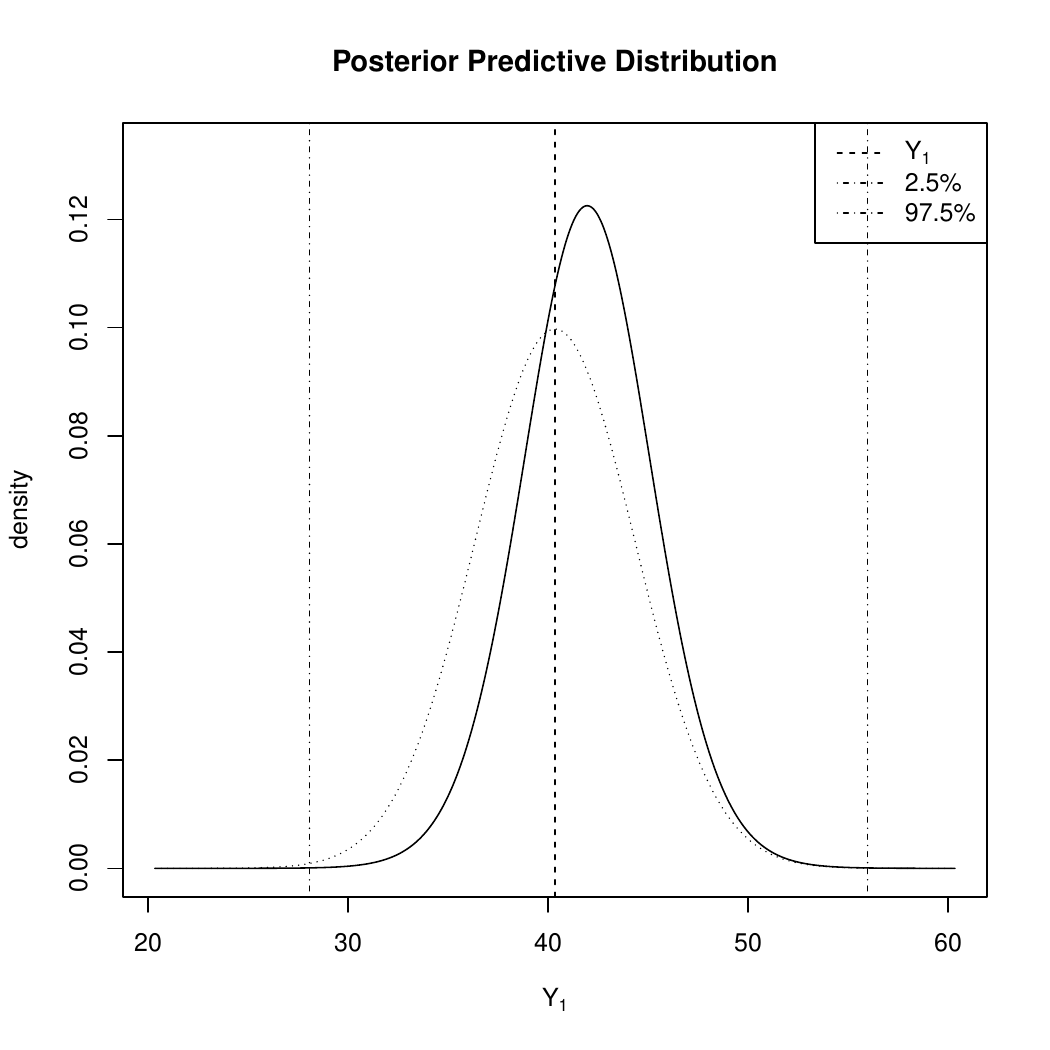}
%\caption*{Scenario I}
\end{minipage}
\centering
\begin{minipage}[b]{0.4\linewidth}
\centering
\includegraphics[width=\textwidth]{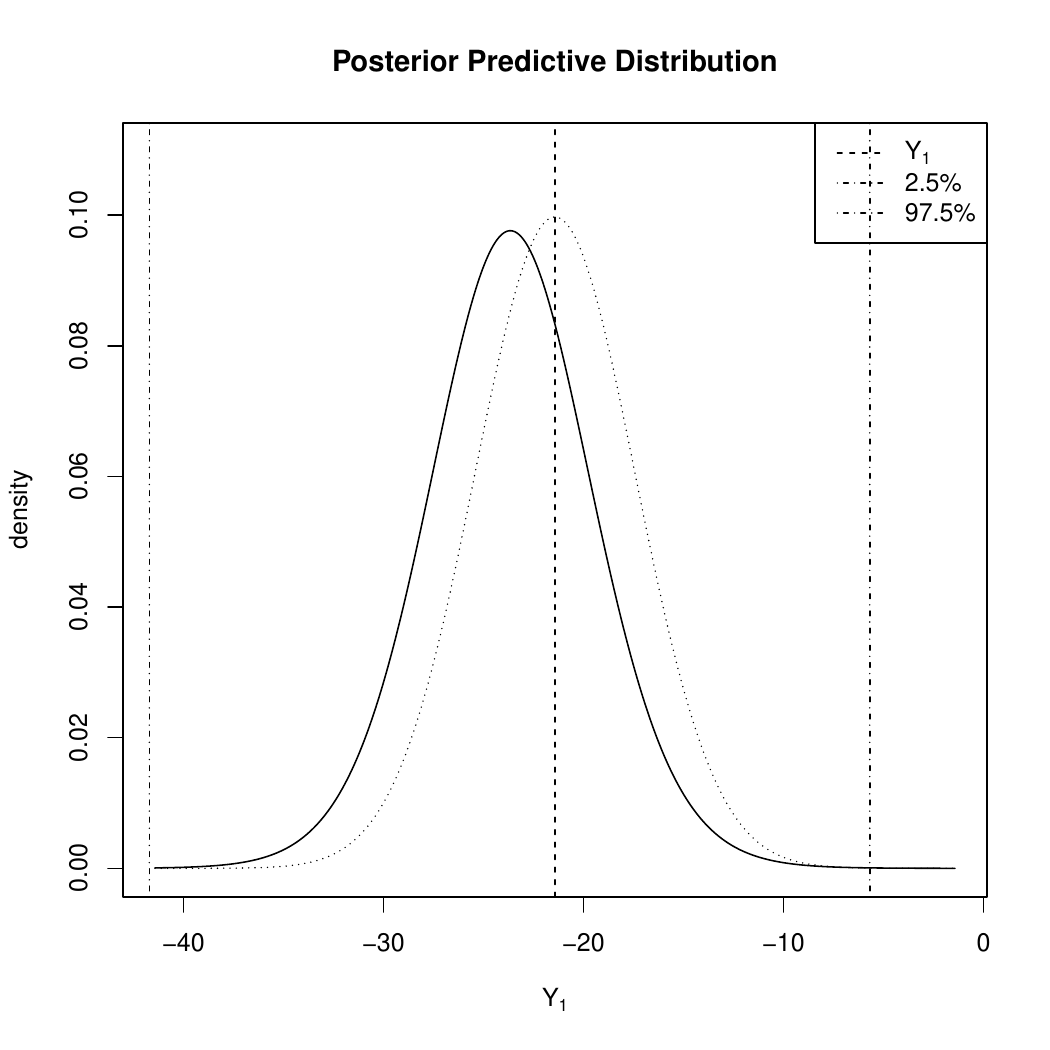}
%\caption*{Scenario II}
\end{minipage}
\begin{minipage}[b]{0.4\linewidth}
\centering
\includegraphics[width=\textwidth]{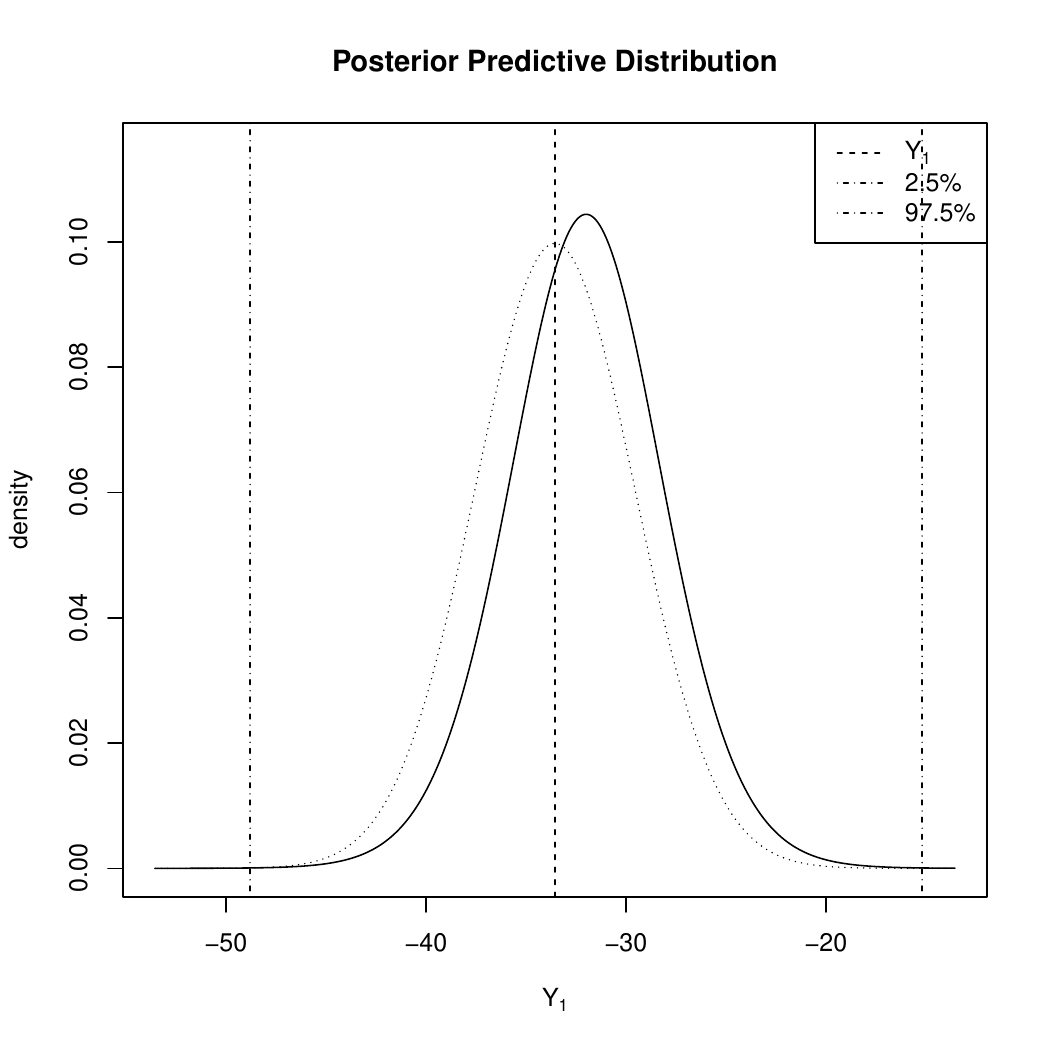}
%\caption*{Scenario III}
\end{minipage}
\begin{minipage}[b]{0.4\linewidth}
\centering
\includegraphics[width=\textwidth]{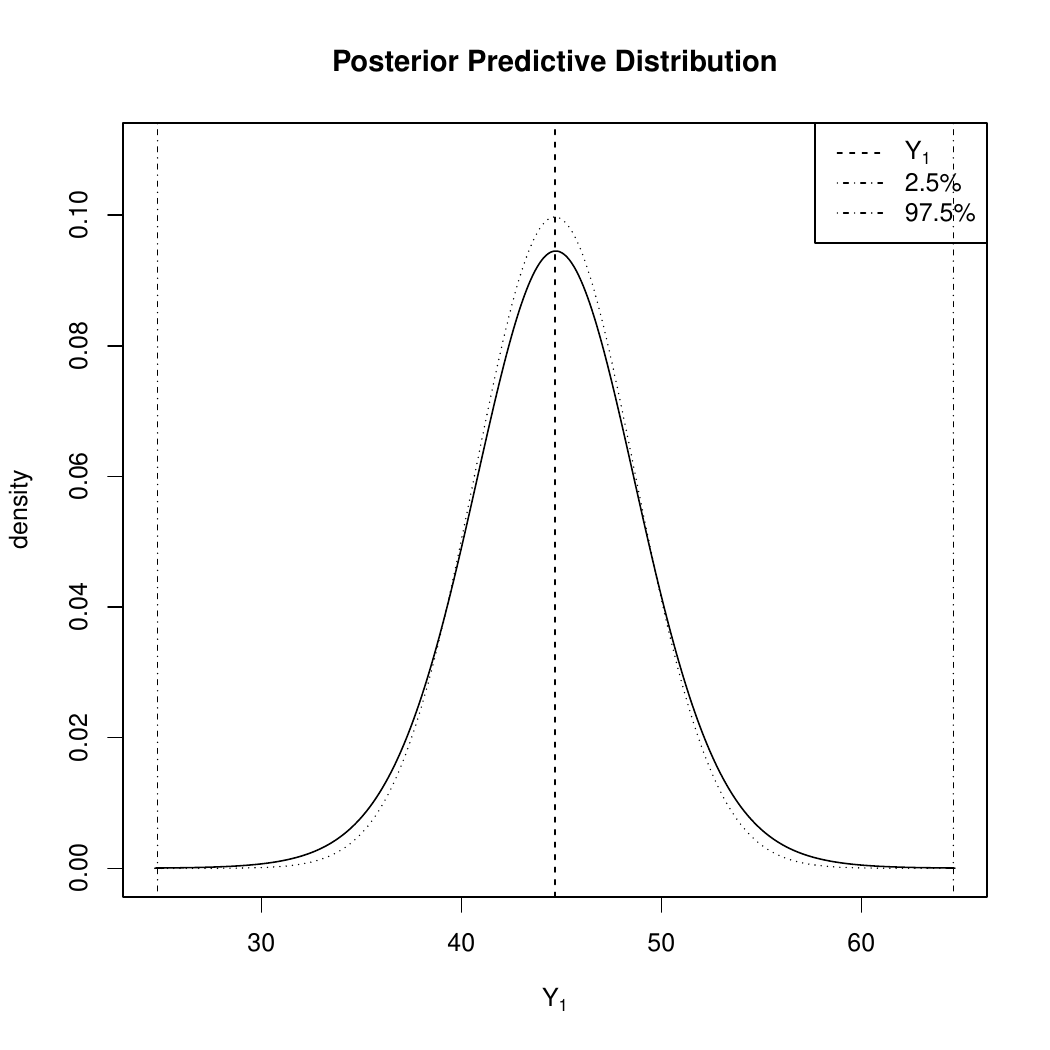}
%\caption*{Scenario IV}
\end{minipage}
\end{figure}
\subsection{Variable Selection}
Based on the posterior consistency of our method, we can also consider variable selection. An empirical posterior credible interval is naturally available for each coefficient. We set a coefficient to 0 when the corresponding 95\% posterior credible interval contains 0 in all the four scenarios and compare the results with those from lasso, spike-and-slab lasso, and horseshoe model. Table~\ref{tab:variable_selection} gives the average size of fitted models and frequency of catching the true model in the 500 replications. Perfect selection gives the number of nonzero coefficients in each model, including the intercept. The bridge prior model significantly  outperforms the others when the true model has half of the coefficients being small or is super sparse, while having a harder time in the dense scenario with small coefficients. Even though not reported here, the bridge prior model has the smallest false discovery rate among all the methods considered.  
\begin{table}[htbp]\small
   \centering

      \caption{Average number of selected variables and frequency finding the true model.}
     \begin{tabular}{lllllllll}
      \toprule
  
   &  \multicolumn{8}{c}{Scenario} \\
 \cmidrule(lr){2-9}
 Method & \multicolumn{2}{c}{I} &\multicolumn{2}{c}{II}&\multicolumn{2}{c}{III}&\multicolumn{2}{c}{IV} \\
       \cmidrule(lr){2-3}\cmidrule(lr){4-5}\cmidrule(lr){6-7}\cmidrule(lr){8-9}
      Perfect Selection & \multicolumn{2}{c}{3} &\multicolumn{2}{c}{11}&\multicolumn{2}{c}{21}&\multicolumn{2}{c}{9} \\
      B.P. & 3 & 500 & 11.030 & 443  & 18.446 & 23  & 9.012 & 493\\ 
Lasso  & 4.280 & 124 & 13.914 & 12 & 20.904 & 452& 20.316 & 0\\ 
SSLasso & 3.762 & 295 & 11.750 & 244  & 20.506 & 298 & 21.322 & 39   \\ 
Horseshoe & 7.138 & 7  & 13.740 & 20 & 20.382 & 261 & 10.354 & 165  \\     
  
  \bottomrule
   \end{tabular}
 %  \caption*{$\dagger$ Standard errors of the medians are estimated via bootstrap with $B =1000$ resamplings on the 1000 mean-squared errors.}
   \label{tab:variable_selection}
\end{table}
\section{Diabetes Data}
\label{sec:diabetes}
\begin{figure}[ht!]
\begin{minipage}[b]{0.32\linewidth}
\centering
\includegraphics[width=\textwidth]{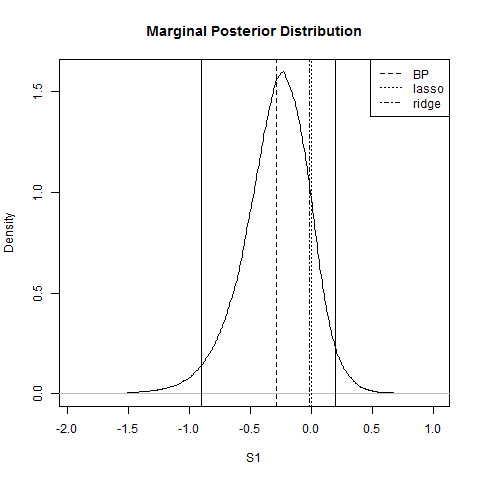}
\end{minipage}
\begin{minipage}[b]{0.32\linewidth}
\centering
\includegraphics[width=\textwidth]{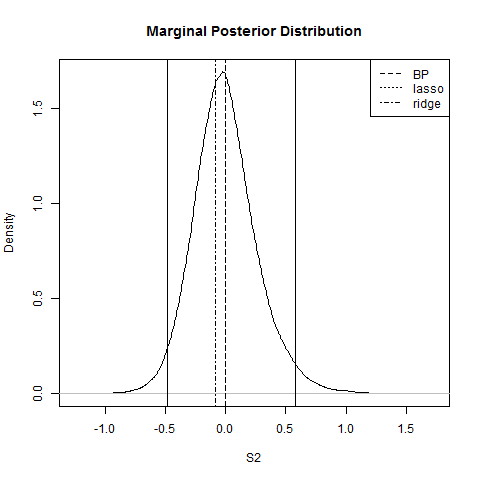}
\end{minipage}
\begin{minipage}[b]{0.32\linewidth}
\centering
\includegraphics[width=\textwidth]{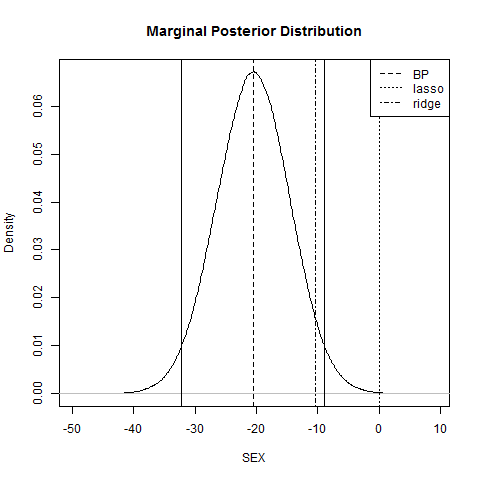}
\end{minipage}
\caption{Marginal posterior distributions for three selected predictors. The dashed line is the marginal posterior mean of each coefficient, dotted line the lasso estimates, and dash-dot line the ridge estimates. An empirical 95\% credible region is marked by solid line for each coefficient as well. }
\label{fig:real}
\end{figure}
The \texttt{Diabetes} data set \citep{efro:etal:2004} contains 10 predictors, 1 response, and 442 observations.  There is a positive correlation between predictor \texttt{S1} and \texttt{S2}. When the correlation between two predictors is close to 1 and both of them tend to be unimportant in the model, some methods may fit a negative coefficient for one predictor and a positive one for the other. For example, a OLS linear model with the \texttt{diabetes} data will fit a coefficient $-1.09$ for \texttt{S1} and $0.75$ for \texttt{S2}.

We applied the fully Bayesian hierarchical model to this data.  The hyperparameters were taken to be $e_1 = f_1 = 1$, $e_2 = 40$, $f_2 = 0.5$, and $e_3=f_3=0.001$ and we used the same uniform prior for $\alpha$.  We used 1e7 MCMC samples.  The left and middle panel of Figure \ref{fig:real} show the empirical marginal posterior distributions for \texttt{S1} and \texttt{S2} respectively. Marked by solid lines, both 95\% credible intervals suggest that \texttt{S1} and \texttt{S2} should not be included in the model. The lasso and ridge solutions are also presented in the figures. An interesting observation is that the effect of gender is significant based the bridge prior model while insignificant on the other two.

\section{A Multivariate Generalization}
\label{sec:multivariate}
There are many possible extensions of the proposed model.  For example, it is natural to consider multivariate versions.  We will briefly consider one of these; the others are somewhat outside of the scope of the current paper and hence are deferred to future work.

Suppose there are $n$ observations and that each observation consists of $m$ responses so that $Y_i$ is an $m$-vector.  Let $X_i$ be an $m \times p$ matrix of covariates.  We assume for $i=1, \ldots, n$
$$
Y_i | \beta, \Sigma \stackrel{ind}{\sim} \text{N}_{m}(X_i \beta, \Sigma) \; .
$$
We then assign priors
$$
\nu(\beta \vert \lambda,\alpha) = \left(\dfrac{\alpha}{2^{1/\alpha+1}\Gamma(1/\alpha)}\right)^{mp}
\left(\prod_{j=1}^{mp}\lambda_j\right)^{1/\alpha}
\exp\left\{-\frac{1}{2}\sum_{j=1}^{mp}\lambda_j|\beta_j|^\alpha\right\}\, ,
$$
$\Sigma \sim \text{Inv Wishart}(\Psi, v)$, $\alpha \sim \text{Uniform}(0.5, 4)$, 
$$
\nu(\lambda_j|\kappa_j, e_1, f_1, e_2, f_2)  = (1-\kappa_j) \text{Gamma}(\lambda_j; e_1, f_1) + \kappa_j \text{Gamma}(\lambda_j; e_2, f_2) \, ,
$$
and $\kappa_j \sim \text{Bernoulli}(.5)$.  A component-wise MCMC algorithm for the resulting posterior is given in Appendix~\ref{app:multivariate}.

We illustrate the performance of the procedure in a simple example.  Suppose $n=100$, $p=21$, and $m=10$.  In each column of $\beta$ we randomly chose 18 of the 20 predictors to be 0, while the remaining two are independently sampled from a $\text{N}(15,3^2)$. The covariance matrix $\Sigma$ is generated from $\Sigma = LL^T$, where $L$ is a sparse lower triangular matrix with 95\% of its elements being 0. 

We implemented the MCMC algorithm to estimate the posterior means $E[\beta|y]$ and $E[\Sigma|y]$. As shown in Figure~\ref{fig:diff_beta}, the difference between true and estimated coefficients matrix $\beta$ are in $[0,0.5]$. Estimation in covariance matrix and precision matrix, as expected, have more small but non-zero cells compared to the truth. Results are shown in Figure~\ref{fig:cov} . This example, in general shows that our procedure is effective at estimating the the true regression coefficients as well as the covariance and precision matrices.  The results here were typical of our other simulations which are not reported here.

\begin{figure}[ht!]
 \centering
  \includegraphics[width=0.7\textwidth]{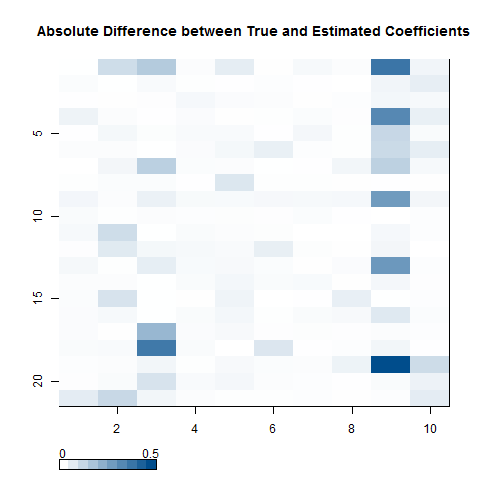}
\caption{Heatmap of $|\beta-\hat{\beta}|$}
 \label{fig:diff_beta}
\end{figure}

\begin{figure}[htbp]
\begin{minipage}[b]{0.45\linewidth}
\centering
\includegraphics[width=\textwidth]{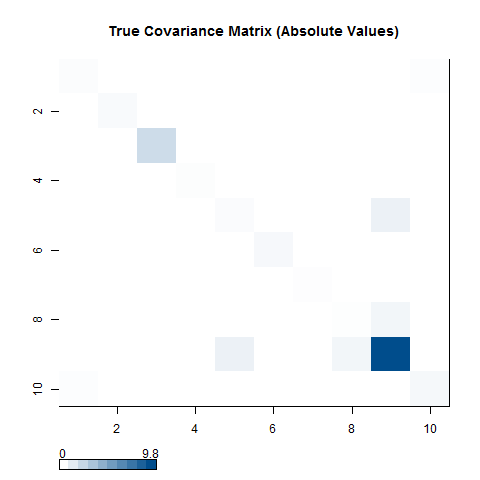}
\end{minipage}
\centering
\begin{minipage}[b]{0.45\linewidth}
\centering
\includegraphics[width=\textwidth]{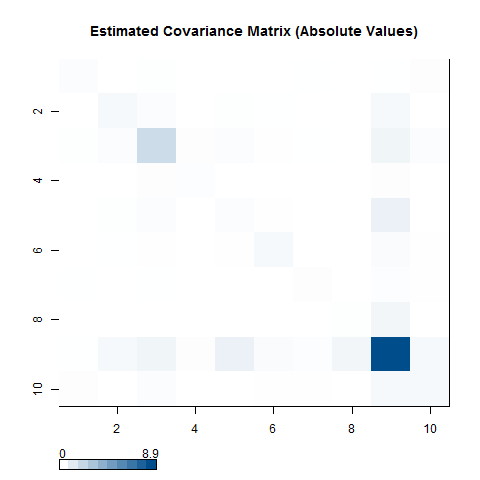}
\end{minipage}
\begin{minipage}[b]{0.45\linewidth}
\centering
\includegraphics[width=\textwidth]{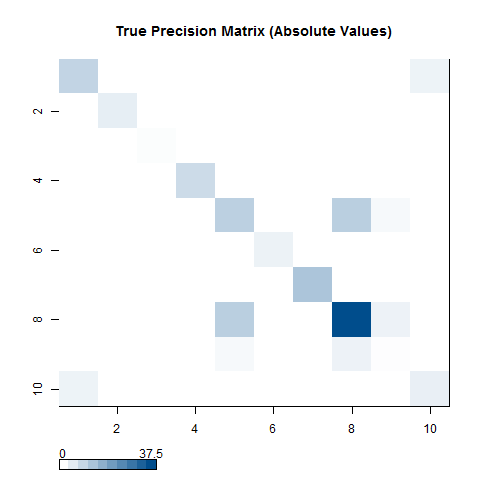}
\end{minipage}
\begin{minipage}[b]{0.45\linewidth}
\centering
\includegraphics[width=\textwidth]{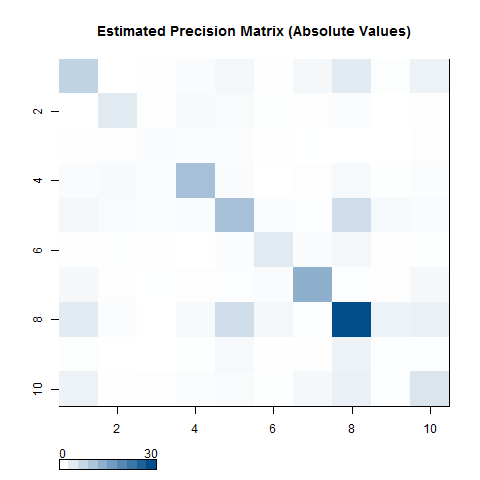}
\end{minipage}
\caption{Covariance matrix and precision matrix}
\label{fig:cov}
\end{figure}

\FloatBarrier

\section{Final Remarks}
\label{sec:remarks}
We proposed a fully Bayesian method for penalized linear regression which incorporates shrinkage-related parameters both at the individual and the group level namely, the $\lambda_i$'s and $\alpha$.  This allows the practitioner to address the uncertainty in the tuning parameters in a principled manner by averaging over the posterior distribution.  Overall, the method has cutting-edge performance in terms of prediction, estimation, and variable selection.

There are several potential directions for future research in this vein.  We considered one possible multivariate generalization and showed that it was effective at estimating covariance parameters.   Other multivariate approaches are certainly possible.  It would also be interesting to consider a broader class of univariate penalized regression approaches which would allow the embedding of lasso, ridge and bridge regression in a framework which also incorporates other penalized regressions such as the elastic net.

\begin{appendix}
\section{Proof of Theorem \ref{thm:shrinkage}}
\label{app:proof:shrinkage}
\begin{proof}
Notice that
\begin{align*}
m(y)=&\frac{1}{k_2-k_1}\int\int \frac{1}{\sqrt{2\pi}}\exp\left(-\frac{(\beta-y)^2}{2}\right)\dfrac{\alpha}{2^{1/\alpha+1}\Gamma(1/\alpha)}\\
&\left\{\frac{1}{2}\frac{f_1^{e_1}}{\Gamma(e_1)}\frac{\Gamma(e_1+1/\alpha)}{\left(|\beta|^\alpha/2+f_1\right)^{e_1+1/\alpha}}+\frac{1}{2}\frac{f_2^{e_2}}{\Gamma(e_2)}\frac{\Gamma(e_2+1/\alpha)}{\left(|\beta|^\alpha/2+f_2\right)^{e_2+1/\alpha}} \right\}d\beta d\alpha
\end{align*}
For ease of computation we assume that $e_1=e_2=e$, $f_1=f_2=f$ and therefore
\begin{align*}
m(y)=&\frac{1}{k_2-k_1}\int\int \frac{1}{\sqrt{2\pi}}\exp\left(-\frac{(\beta-y)^2}{2}\right)\dfrac{\alpha}{2^{1/\alpha+1}\Gamma(1/\alpha)}\\
&\frac{f^{e}}{\Gamma(e)}\frac{\Gamma(e+1/\alpha)}{\left(|\beta|^\alpha/2+f\right)^{e+1/\alpha}}d\beta d\alpha
\end{align*}
and
\begin{align*}
m'(y)=&\frac{1}{k_2-k_1}\int\int \frac{1}{\sqrt{2\pi}}(\beta-y)\exp\left(-\frac{(\beta-y)^2}{2}\right)\dfrac{\alpha}{2^{1/\alpha+1}\Gamma(1/\alpha)}\\
&\frac{f^{e}}{\Gamma(e)}\frac{\Gamma(e+1/\alpha)}{\left(|\beta|^\alpha/2+f\right)^{e+1/\alpha}}d\beta d\alpha \, .
\end{align*}
Set 
\begin{align*}
g(y) =& \int \exp\left(-\frac{(\beta-y)^2}{2}\right)\frac{1}{\left(|\beta|^\alpha/2+f\right)^{e+1/\alpha}}d\beta \\
h(y) =& \int (\beta-y)\exp\left(-\frac{(\beta-y)^2}{2}\right)\frac{1}{\left(|\beta|^\alpha/2+f\right)^{e+1/\alpha}}d\beta \, .
\end{align*}
Let $t=\beta-y$ so that
\begin{align*}
g(y) =& \int \exp\left(-\frac{t^2}{2}\right)\frac{1}{\left(|t+y|^\alpha/2+f\right)^{e+1/\alpha}}dt \\
h(y) =& \int t\exp\left(-\frac{t^2}{2}\right)\frac{1}{\left(|t+y|^\alpha/2+f\right)^{e+1/\alpha}}dt 
\end{align*}
while if $t = -t$ we then have 
\begin{align*}
g(y) =& \int \exp\left(-\frac{t^2}{2}\right)\frac{1}{\left(|y-t|^\alpha/2+f\right)^{e+1/\alpha}}dt \\
h(y) =& -\int t\exp\left(-\frac{t^2}{2}\right)\frac{1}{\left(|y-t|^\alpha/2+f\right)^{e+1/\alpha}}dt \, .
\end{align*}
Hence 
\begin{align*}
2g(y) =& \int \exp\left(-\frac{t^2}{2}\right)\left[\left(|y+t|^\alpha/2+f\right)^{-e-1/\alpha}+\left(|t-y|^\alpha/2+f\right)^{-e-1/\alpha}\right]dt \\
2h(y) =& \int t\exp\left(-\frac{t^2}{2}\right)\left[\left(|y+t|^\alpha/2+f\right)^{-e-1/\alpha}-\left(|t-y|^\alpha/2+f\right)^{-e-1/\alpha}\right]dt 
\end{align*}
Notice that both functions under the integral sign are even.  Thus
\begin{align*}
g(y) =& \int_0^\infty \exp\left(-\frac{t^2}{2}\right)\left[\left(|y+t|^\alpha/2+f\right)^{-e-1/\alpha}+\left(|t-y|^\alpha/2+f\right)^{-e-1/\alpha}\right]dt \\
h(y) =& \int_0^\infty t\exp\left(-\frac{t^2}{2}\right)\left[\left(|y+t|^\alpha/2+f\right)^{-e-1/\alpha}-\left(|t-y|^\alpha/2+f\right)^{-e-1/\alpha}\right] dt \, .
\end{align*}
Suppose $y > 0$ (a nearly identical proof will hold with $y < 0$).  Then 
\begin{align}
g(y)>& \int_0^\infty \exp\left(-\frac{t^2}{2}\right)\left(|t-y|^\alpha/2+f\right)^{-e-1/\alpha}dt \nonumber\\
>& \int_0^{y/2} \exp\left(-\frac{t^2}{2}\right)\left(|t-y|^\alpha/2+f\right)^{-e-1/\alpha}dt\nonumber\\
=& \int_0^{y/2} \exp\left(-\frac{t^2}{2}\right)\left((y-t)^\alpha/2+f\right)^{-e-1/\alpha}dt\nonumber\\
>& \int_0^{y/2} \exp\left(-\frac{t^2}{2}\right)\left(y^\alpha/2+f\right)^{-e-1/\alpha}dt\nonumber\\
=&C_1\cdot\left(y^\alpha/2+f\right)^{-e-1/\alpha},\label{fun_g}
\end{align}
where $C_1 < \frac{\sqrt{2\pi}}{2}$. Next consider 
\begin{align*}
-h(y) =& \int_0^\infty t\exp\left(-\frac{t^2}{2}\right)\left[\left(|t-y|^\alpha/2+f\right)^{-e-1/\alpha}-\left(|y+t|^\alpha/2+f\right)^{-e-1/\alpha}\right]dt\\
=& h_1(y) + h_2(y),
\end{align*}
where 
\begin{align*}
h_1(y) =& \int_0^{y/2} t\exp\left(-\frac{t^2}{2}\right)\left[\left((y-t)^\alpha/2+f\right)^{-e-1/\alpha}-\left((y+t)^\alpha/2+f\right)^{-e-1/\alpha}\right]dt\\
h_2(y)=& \int_{y/2}^\infty t\exp\left(-\frac{t^2}{2}\right)\left[\left(|t-y|^\alpha/2+f\right)^{-e-1/\alpha}-\left(|y+t|^\alpha/2+f\right)^{-e-1/\alpha}\right]dt \, .
\end{align*}
Set
$$S(t) = \left((y-t)^\alpha/2+f\right)^{-e-1/\alpha}-\left((y+t)^\alpha/2+f\right)^{-e-1/\alpha}$$
We want to show that, when $0<t<y/2$,
$$S(t) < V(t) = 4\alpha (e+1/\alpha) \left((y/2)^\alpha/2+f\right)^{-e-1/\alpha} \frac{t}{y} \, . $$
First notice that $S(0) = V(0) = 0$ so that we only need to show that $S'(t) < V'(t)$. Consider
\begin{align*}
S'(t) =& \frac{\alpha (e+1/\alpha)}{2}(y-t)^{\alpha-1}\left(\frac{(y-t)^\alpha}{2}+f\right)^{-(e+1/\alpha+1)}\\
&+ \frac{\alpha (e+1/\alpha)}{2}(y+t)^{\alpha-1}\left(\frac{(y+t)^\alpha}{2}+f\right)^{-(e+1/\alpha+1)}\\
V'(t) =& \left((y/2)^\alpha/2+f\right)^{-e-1/\alpha}\cdot \frac{1}{y}\cdot 4\alpha (e+1/\alpha)
\end{align*}
Notice that  
\begin{align*}
f(x) = x^{\alpha-1}\left(\frac{x^\alpha}{2}+f\right)^{-(e+1/\alpha+1)} = x^{-(2+\alpha e)} \left( \frac{1}{2} + \frac{f}{x^{\alpha}}\right)^{-(e + 1 + 1/\alpha)}
\end{align*}
%With some algebra, it is equivalent to    
%\begin{align*}
% f(x)= x^{-1}\left(\frac{x^{\alpha(e+1/\alpha)/(e+1/\alpha+1)}}{2}+\frac{f}{x^{\alpha/(e+1/\alpha+1)}}\right)^{-(e+1/\alpha+1)}
%\end{align*}
is a decreasing function when $x$ is large and $0<k_1<\alpha<k_2\leq4$. Thus, when $y$ is large and $0<t<y/2$, we have $y-t>y/2$ and $y+t>y>y/2$.  Therefore,
\begin{align*}
yS'(t)=& \frac{\alpha (e+1/\alpha)}{2}y(y-t)^{\alpha-1}\left(\frac{(y-t)^\alpha}{2}+f\right)^{-(e+1/\alpha+1)}\\
&+ \frac{\alpha (e+1/\alpha)}{2}y(y+t)^{\alpha-1}\left(\frac{(y+t)^\alpha}{2}+f\right)^{-(e+1/\alpha+1)}\\
<& \alpha (e+1/\alpha)y(y/2)^{\alpha-1}\left(\frac{(y/2)^\alpha}{2}+f\right)^{-(e+1/\alpha+1)}\\
< &\alpha (e+1/\alpha)y(y/2)^{\alpha-1}\left(\frac{(y/2)^\alpha}{2}\right)^{-1}\left(\frac{(y/2)^\alpha}{2}+f\right)^{-(e+1/\alpha)}\\
= & 4\alpha (e+1/\alpha)\left((y/2)^\alpha/2+f\right)^{-e-1/\alpha}.
\end{align*}
which gives us $S'(t)<V'(t)$. We now have $S(t)<V(t)$. Then 
\begin{align*}
h_1(y)=& \int_0^{y/2}  t\exp\left(-\frac{t^2}{2}\right) S(t) dt\\<& \int_0^{y/2} t^2\exp\left(-\frac{t^2}{2}\right)\left((y/2)^\alpha/2+f\right)^{-e-1/\alpha}\cdot \frac{1}{y}\cdot 4\alpha (e+1/\alpha)dt\\
<&\left((y/2)^\alpha/2+f\right)^{-e-1/\alpha}\cdot \frac{1}{y}\cdot 4\alpha (e+1/\alpha)\int_{-\infty}^{\infty} t^2\exp\left(-\frac{t^2}{2}\right)dt\\
<&\left((y/2)^\alpha/2+f\right)^{-e-1/\alpha}\cdot \frac{1}{y}\cdot 4\alpha (e+1/\alpha)\cdot 2\sqrt{2\pi}
\end{align*}
and
\begin{align*}
h_2(y)=& \int_{y/2}^\infty t\exp\left(-\frac{t^2}{2}\right)\left[\left(|t-y|^\alpha/2+f\right)^{-e-1/\alpha}-\left(|y+t|^\alpha/2+f\right)^{-e-1/\alpha}\right]dt\\
<& \int_{y/2}^\infty t\exp\left(-\frac{t^2}{2}\right)\left(|t-y|^\alpha/2+f\right)^{-e-1/\alpha}dt\\
<& f^{-e-1/\alpha}\int_{y/2}^\infty t\exp\left(-\frac{t^2}{2}\right)dt\\
=&f^{-e-1/\alpha}\exp\left(-\frac{y^2}{8}\right) \, .
\end{align*}
Therefore, 
\begin{align}
-h(y)=&h_1(y)+h_2(y)\nonumber\\
<&\left((y/2)^\alpha/2+f\right)^{-e-1/\alpha}\cdot \frac{1}{y}\cdot 4\alpha (e+1/\alpha) \cdot 2\sqrt{2\pi}+f^{-e-1/\alpha}\exp\left(-\frac{y^2}{8}\right) \label{fun_h}
\end{align}
By \eqref{fun_g} and \eqref{fun_h} we have
\begin{align*}
  0<&-\frac{m'(y)}{m(y)}=\frac{\int-h(y)\dfrac{\alpha\Gamma(e+1/\alpha)}{2^{1/\alpha}\Gamma(1/\alpha)}d\alpha}{   \int g(y)\dfrac{\alpha\Gamma(e+1/\alpha)}{2^{1/\alpha}\Gamma(1/\alpha)} d\alpha}\\
 <&\frac{1}{y}\frac{\int\dfrac{\alpha\Gamma(e+1/\alpha)}{2^{1/\alpha+1}\Gamma(1/\alpha)}\left[\left((y/2)^\alpha/2+f\right)^{-e-1/\alpha}8\alpha(e+1/\alpha)\right]d\alpha}{\int\dfrac{\alpha\Gamma(e+1/\alpha)}{2^{1/\alpha+1}\Gamma(1/\alpha)}C_1\left(y^\alpha/2+f\right)^{-e-1/\alpha}(2\pi)^{-1/2}d\alpha}\\
 &+\frac{\int\dfrac{\alpha\Gamma(e+1/\alpha)}{2^{1/\alpha+1}\Gamma(1/\alpha)}\left[f^{-e-1/\alpha}(2\pi)^{-1/2}\exp\left(-\frac{y^2}{8}\right)\right]d\alpha}{\int\dfrac{\alpha\Gamma(e+1/\alpha)}{2^{1/\alpha+1}\Gamma(1/\alpha)}C_1\left(y^\alpha/2+f\right)^{-e-1/\alpha}(2\pi)^{-1/2}d\alpha}\\
 \equiv& \frac{1}{y}A_1 + A_2
 \end{align*}
 Notice that $A_1$ is bounded by a constant since
 \begin{align*}
 A_1 =& \frac{\int\dfrac{\alpha\Gamma(e+1/\alpha)}{2^{1/\alpha+1}\Gamma(1/\alpha)}\left[\left((y/2)^\alpha/2+f\right)^{-e-1/\alpha}8\alpha(e+1/\alpha)\right]d\alpha}{\int\dfrac{\alpha\Gamma(e+1/\alpha)}{2^{1/\alpha+1}\Gamma(1/\alpha)}C_1\left(y^\alpha/2+f\right)^{-e-1/\alpha}(2\pi)^{-1/2}d\alpha}\\
 \leq & \frac{\int\dfrac{k_2\Gamma(e+1/k_1)}{2^{1/k_2+1}\Gamma(1/k_2)}\left[\left((y/2)^\alpha/2+f\right)^{-e-1/\alpha}8k_2(e+1/k_1)\right]d\alpha}{\int\dfrac{k_1\Gamma(e+1/k_2)}{2^{1/k_1+1}\Gamma(1/k_1)}C_1\left(y^\alpha/2+f\right)^{-e-1/\alpha}(2\pi)^{-1/2}d\alpha}\\
 \equiv& C_{e,f,k_1,k_2}^{(0)}\frac{\int\left((y/2)^\alpha/2+f\right)^{-e-1/\alpha}d\alpha}{\int\left(y^\alpha/2+f\right)^{-e-1/\alpha}d\alpha}\\
 < &C_{e,f,k_1,k_2}^{(0)}\frac{\int\left(y^\alpha/2+f\right)^{-e-1/\alpha}d\alpha}{\int\left(y^\alpha/2+f\right)^{-e-1/\alpha}d\alpha}=C_{e,f,k_1,k_2}^{(0)}\\
 \end{align*}
Notice that, because of the term $\exp(-\frac{y^2}{8})$, $A_2$ is a higher order term of $A_1$ when $y$ goes to infinity. Then we can write
 \begin{align*}
\frac{1}{y}A_1+A_2=&\frac{1}{y}C^{(1)}_{e,f,k_1,k_2} + o(\frac{1}{y})C^{(2)}_{e,f,k_1,k_2}
\sim O(1/y), 
\end{align*}
since $C^{(1)}_{e,f,k_1,k_2}<C_{e,f,k_1,k_2}^{(0)}$ and $C^{(2)}_{e,f,k_1,k_2}$ are constants depending on the choice of $k_1,k_2,e,f$. Therefore 
$\lim_{y\rightarrow \infty}\frac{m'(y)}{m(y)} =0$. 
\end{proof}

\section{Proof of Theorem \ref{thm:conditional} and \ref{thm:marginal}}
\label{app:proofs}
We begin with a preliminary result that will be used in both proofs.

\begin{lemma}
\label{lem:ineq}
Let $A_n$ denote the subset of nonzero entries in $\beta_n^0$ and $0<\Delta\leq \varepsilon^2\Lambda^2_{\text{min}}/(48\Lambda^2_{\text{max}})$ and $\rho >0$.  Then
\begin{align}
\label{eq:ineq}
&\nu_n\left\{\beta_n:||\beta_n-\beta_n^0||<\Delta/n^{\rho/2}\,|\,\theta_n\right\} \geq \nonumber\\
&\left\{ \frac{\Delta}{\sqrt{p_n}n^{\rho/2}} \frac{\alpha_n(\gamma\lambda_{nj})^{1/\alpha_n}}{2^{1/\alpha_n} \Gamma(1/\alpha_n)} \exp \left(-\frac{\gamma\lambda_{nj}}{2}\left(\left(\frac{\Delta}{\sqrt{p_n}n^{\rho/2}}\right)^{\alpha_n} + \sup_{j\in A_n}|\beta_{nj}^0|^{\alpha_n}\right) \right)\right\}^{q_n} \nonumber\\
&\times \left(1 - \frac{p_nn^\rho \Gamma(3/\alpha_n)(\gamma\lambda_{nj})^{-2/\alpha_n}4^{1/\alpha_n}}{\Gamma(1/\alpha_n)\Delta^2} \right)\, .
\end{align}
\end{lemma}

\begin{proof}
Notice that
\begin{align*}
&\cap_{j\in A_n} \left\{\beta_{nj}: |\beta_{nj}-\beta_{nj}^0|<\frac{\Delta}{\sqrt{p_n}n^{\rho/2}} \right\} \cap \left\{\beta_{n,j\not\in A_n}: \sum_{j\not\in A_n} \beta_{nj}^2 < \frac{(p_n-q_n)\Delta^2}{p_nn^\rho} \right\} \\
&\subseteq \left\{\beta_n : \sum_{j\in A_n} (\beta_{nj}-\beta_{nj}^0)^2+\sum_{j\not\in A_n} \beta_{nj}^2 < \frac{\Delta^2}{n^\rho} \right\} \, .
\end{align*}
Then
\begin{align*}
& \nu_n\left\{\beta_n:||\beta_n-\beta_n^0||<\Delta/n^{\rho/2}\, |\, \theta_n \right\} = \nu_n\left\{\beta_n : \sum_{j\in A_n}(\beta_{nj}-\beta_{nj}^0)^2+\sum_{j\not\in A_n} \beta_{nj}^2 < \Delta^2 / n^\rho\, |\, \theta_n \right\}\\
&\geq\nu_n\left\{\cap_{j\in A_n}\left(\beta_{nj}: |\beta_{nj}-\beta_{nj}^0| < \frac{\Delta}{\sqrt{p_n}n^{\rho/2}}\right)\cap\left(\beta_{n,j\not\in A_n}: \sum_{j\not\in A_n} \beta_{nj}^2 < \frac{(p_n-q_n)\Delta^2}{p_nn^\rho}\right)\,|\,\theta_n\right\} \; .
\end{align*}
and by conditional independence
\begin{align}
\label{Ccite1}
%&\nu_n\left\{\beta_n:||\beta_n-\beta_n^0|| <\Delta/n^{\rho/2}\,|\,\theta_n\right\}\nonumber\\
%&\geq\nu_n\left\{\cap_{j\in\mathcal{A}}\left(\beta_{nj}: |\beta_{nj}-\beta_{nj}^0|<\frac{\Delta}{\sqrt{p_n}n^{\rho/2}}\right)\cap\left(\beta_{n,j\not\in\mathcal{A}_n}: \sum_{j\not\in\mathcal{A}_n} \beta_{nj}^2 < \frac{(p_n-q_n)\Delta^2}{p_nn^\rho}\right)\,|\,\theta_n\right\}\nonumber\\
&\geq \prod_{j\in A_n} \nu_n\left\{\beta_{nj}: |\beta_{nj}-\beta_{nj}^0| < \frac{\Delta}{\sqrt{p_n}n^{\rho/2}}\,|\, \theta_n\right\} \nu_n\left\{\beta_{n,j\not\in A_n}: \sum_{j\not\in A_n} \beta_{nj}^2 < \frac{(p_n-q_n)\Delta^2}{p_nn^\rho}\,| \,\theta_n\right\}\nonumber\\
&=\prod_{j\in A_n} \nu_n\left\{ \beta_{nj}: |\beta_{nj}-\beta_{nj}^0| < \frac{\Delta}{\sqrt{p_n}n^{\rho/2}}\, | \, \theta_n\right\}\left[1-\nu_n\left\{ \beta_{n,j\not\in A_n} : \sum_{j\not\in A_n} \beta_{nj}^2 \geq \frac{(p_n-q_n)\Delta^2}{p_nn^\rho}\,|\,\theta_n\right\}\right]\nonumber\\
&\geq\prod_{j\in A_n}\nu_n\left\{ \beta_{nj}: |\beta_{nj}-\beta_{nj}^0|<\frac{\Delta}{\sqrt{p_n}n^{\rho/2}}\,|\,\theta_n\right\} \left[1 - \frac{p_nn^\rho E \left(\sum_{j\not\in A_n}\beta_{nj}^2\,|\,\theta_n\right)}{(p_n-q_n)\Delta^2}\right]\, .
\end{align}
For $j \in A_n$, we have 
\begin{align}
\label{the_other}
&\nu_n\left\{\beta_{nj}: |\beta_{nj}-\beta_{nj}^0|<\frac{\Delta}{\sqrt{p_n}n^{\rho/2}}\,| \, \theta_n\right\}\nonumber\\
=& \int_{|\beta_{nj}-\beta_{nj}^0|<\frac{\Delta}{\sqrt{p_n}n^{\rho/2}}} \frac{\alpha_n(\gamma\lambda_{nj})^{1/\alpha_n}}{2^{1/\alpha_n+1}\Gamma(1/\alpha_n)}\exp\left\{-\frac{\gamma\lambda_{nj}}{2}|\beta_{nj}|^{\alpha_n}\right\} \, d\beta_{nj} \nonumber\\
\geq& \int_{|\beta_{nj}-\beta_{nj}^0|<\frac{\Delta}{\sqrt{p_n}n^{\rho/2}}}\frac{\alpha_n(\gamma\lambda_{nj})^{1/\alpha_n}}{2^{1/\alpha_n+1}\Gamma(1/\alpha_n)}\exp\left\{-\frac{\gamma\lambda_{nj}}{2}\left(|\beta_{nj}-\beta_{nj}^0|^{\alpha_n}+ |\beta_{nj}^0|^{\alpha_n}\right)\right\} \, d\beta_{nj}\nonumber\\
\geq& \frac{\Delta}{\sqrt{p_n}n^{\rho/2}}\frac{\alpha_n(\gamma\lambda_{nj})^{1/\alpha_n}}{2^{1/\alpha_n}\Gamma(1/\alpha_n)}\exp\left\{-\frac{\gamma\lambda_{nj}}{2} \left(\left(\frac{\Delta}{\sqrt{p_n}n^{\rho/2}}\right)^{\alpha_n}+ \sup_{j\in A_n}|\beta_{nj}^0|^{\alpha_n}\right) \right\}\, .
\end{align}
Since
\begin{align*}
E(\beta_{nj}^2|\Theta_n)= \frac{\Gamma(3/\alpha_n)}{\Gamma(1/\alpha_n)}(\gamma\lambda_{nj})^{-2/\alpha_n}4^{1/\alpha_n}.
\end{align*}
combining \eqref{Ccite1} and \eqref{the_other} yields the claim.
\end{proof}

\begin{proof}[Proof of Theorem~\ref{thm:conditional}] 
All we need to do is show that there exists $N$ such that for $n \ge N$
\begin{align*} \nu_n\left\{\beta_n: ||\beta_n-\beta_n^0||<\frac{\Delta}{n^{\rho/2}}\,|\,\theta_n\right\} > \exp(-dn)\, .  \end{align*} 
The result would then follow directly from Theorem 1 in \cite{arma:etal:2013b}.

Consider \eqref{eq:ineq}. If $\gamma\lambda_{nj} = (C\sqrt{p_n}n^{\rho/2}\log n)^{\alpha_n}$ for finite $C> \sqrt{\frac{\Gamma(3/\alpha_n) 4^{1/\alpha_n}}{\Gamma(1/\alpha_n)\Delta^2\log^2n}}$, then 
\[ 1-\frac{p_nn^\rho\Gamma(3/\alpha_n)(\gamma \lambda_{nj})^{-2/\alpha_n}4^{1/\alpha_n}}{\Gamma(1/\alpha_n)\Delta^2} >0\, .  
\] 
Now take the negative logarithm of both sides of \eqref{eq:ineq} to obtain 
\begin{align*}
-\log\nu_n\left\{\beta_n : || \beta_n-\beta_n^0 || < \Delta/n^{\rho/2} \, | \, \theta_n\right\} \leq &-q_n\log\Delta - q_n\log\frac{C\alpha_n\log n}{2^{1/\alpha_n} \Gamma(1/\alpha_n)}\\                                                                                          &+q_nC^{\alpha_n}(\log n)^{\alpha_n} \Delta^{\alpha_n}/2\\
  &+ q_n \left(C\sqrt{p_n}n^{\rho/2}\log n \right)^{\alpha_n}\sup_{j \in \mathcal{A}_n}|\beta_{nj}|^{\alpha_n}/2\\                                                             
&-\log \left(1-\frac{4^{1/\alpha_n} \Gamma(3/\alpha_n)}{C^2\Gamma(1/\alpha_n)(\log n)^2 \Delta^2} \right) \, .  
\end{align*} 
By assumption $q_n=o\{n^{1-\rho}/(p_n\log^2 n)\}$ as $n \to \infty$ for $\rho\in (0,1)$ so that also $q_n =o(n)$.  Thus the first, second, and the fifth term on the right hand side are $o(n)$. Consider the third and fourth terms.  As $n \rightarrow \infty$, \begin{align*}
 0 &\leftarrow \frac{q_n}{n^{1-\rho}/(p_n\log^2 n)}\\
 & = \frac{q_n\left( C\sqrt{p_n}n^{\rho/2}\log n \right)^{\alpha_n}}{n^{1-\rho}/(p_n\log^2 n) \cdot \left( C\sqrt{p_n}n^{\rho/2}\log n \right)^{\alpha_n}}\\
&\geq \frac{q_n\left(C\sqrt{p_n}n^{\rho/2}\log n \right)^{\alpha_n}}{n^{1-\rho}/(p_n\log^2 n) \cdot \left(C\sqrt{p_n}n^{\rho/2}\log n \right)^2}\\
 &= \frac{q_n \left(C\sqrt{p_n}n^{\rho/2}\log n \right)^{\alpha_n}}{C^2n} \geq 0 \, .  \end{align*} 
Thus, as $n \to \infty$, $q_n(C\sqrt{p_n}n^{\rho/2}\log n )^{\alpha_n} = o(n)$ and $q_n(\log n)^{\alpha_n} = o(n)$.  Therefore the third and fourth terms are also $o(n)$. Then for all $d>0$, as $n \rightarrow \infty$, 
\begin{align*} -\log\nu_n\left\{\beta_n:||\beta_n-\beta_n^0||<\Delta/n^{\rho/2}\,| \, \theta_n \right\} < dn 
\end{align*} 
and the result follows.
\end{proof}

\begin{proof}[Proof of Theorem~\ref{thm:marginal}]
We need to show that there exists $N$ such that for $n \ge N$
\begin{align*} 
\nu_n\left\{\beta_n: ||\beta_n-\beta_n^0||< \Delta / n^{\rho/2} \right\} > \exp(-dn)\, .  
\end{align*} 
The result would then follow directly from Theorem 1 in \cite{arma:etal:2013b}.

Let $B_n =  \{\beta_n: ||\beta_n-\beta_n^0||< \Delta / n^{\rho/2} \}$. Then, since $\nu_n (d\theta_n)$ is a proper prior distribution,
\begin{align*}
\nu_n (B_n) = \int_{B_n} \int \nu_n(d\beta_n | \theta_n) \, \nu_n (d\theta_n)
 = \int \nu_n(B_n | \theta_n) \, \nu_n (d\theta_n)
 \ge \inf_{\theta_n} \nu_n (B_n | \theta_n) \; .
\end{align*}
We will show that
there exists $N$ such that for $n \ge N$
\begin{align*} 
\inf_{\theta_n} \nu_n(B_n \, | \, \theta_n ) > \exp(-dn)\, .  
\end{align*} 
From \eqref{eq:ineq}  we have 
\begin{align}
\label{Ecite2}
\nu_n (B_n \,|\,\theta_n) & \geq 
\left\{ \frac{\Delta}{\sqrt{p_n}n^{\rho/2}} \frac{\alpha_n(\gamma\lambda_{nj})^{1/\alpha_n}}{2^{1/\alpha_n} \Gamma(1/\alpha_n)} \exp \left(-\frac{\gamma\lambda_{nj}}{2}\left(\left(\frac{\Delta}{\sqrt{p_n}n^{\rho/2}}\right)^{\alpha_n} + \sup_{j\in A_n}|\beta_{nj}^0|^{\alpha_n}\right) \right)\right\}^{q_n} \nonumber\\
&\times \left(1 - \frac{p_nn^\rho \Gamma(3/\alpha_n)(\gamma\lambda_{nj})^{-2/\alpha_n}4^{1/\alpha_n}}{\Gamma(1/\alpha_n)\Delta^2} \right)\, .
\end{align}
Let $\gamma\lambda_{nj}=(C_n\sqrt{p_n}n^{\rho/2}\log n)^{\alpha_n}$ where
\[
C_n > \sqrt{\frac{\Gamma(3/\alpha_n)4^{1/\alpha_n}}{\Gamma(1/\alpha_n)\Delta^2\log^2n}} \rightarrow 0~~~~~~\text{as}~~n \to \infty\; .
\]
Then 
\[
1 -\frac{p_nn^\rho\Gamma(3/\alpha_n)(\gamma \lambda_{nj})^{-2/\alpha_n}4^{1/\alpha_n}}{\Gamma(1/\alpha_n) \Delta^2}  > 0 \; .
\]
Taking the negative logarithm of both sides of \eqref{Ecite2} we obtain 
\begin{align*}
\sup_{\theta_n}\left[ -\log \nu_n(B_n | \theta_n) \right] \leq &\sup_{\theta_n}\left[ -q_n\log\Delta - q_n\log\frac{C_n \alpha_n\log n}{2^{1/\alpha_n}\Gamma(1/\alpha_n)}\right. \\
&+q_nC_n^{\alpha_n}(\log n)^{\alpha_n} \Delta^{\alpha_n}/2\\
&+ q_n\left(C_n\sqrt{p_n}n^{\rho/2}\log n \right)^{\alpha_n}\sup_{j \in A_n}|\beta_{nj}|^{\alpha_n}/2\\
&\left. -\log \left(1-\frac{4^{1/\alpha_n}\Gamma(3/\alpha_n)}{C_n^2 \Gamma(1/\alpha_n)(\log n)^2\Delta^2} \right) \right] \, .
\end{align*}
Since $C^2 =o(n)$ by assumption, the rest of the proof is the same as the last part of the proof for Theorem~\ref{thm:conditional} and hence is omitted.
\end{proof}

%%%%%%%%%%%%%
%\begin{comment}
%%%%%%%%%%%%%
\section{MCMC Algorithm for Section~\ref{sec:multivariate}}
\label{app:multivariate}
Suppose that we have $n$ observations and that in each observation we have $m$ responses. Responses are dependent in each observation but independent from different observations. Besides, there are $p$ predictors. We then have the model for one observation, with responses centered and predictors standardized,  
$$Y_{i} = X_{i}\beta+\varepsilon_{i}, $$
where $Y_{i}, \varepsilon_{i} $ are both $1\times m$ vectors, $X_{i}$ a $1\times p$ vector, $\beta$ a $p \times m $ matrix, and $\varepsilon_{i} \sim N(0, \Sigma)$.\\
Next the likelihood is 
\begin{align*}
f(Y|X, \beta, \Sigma) = (2\pi)^{-\frac{n}{2}}|\Sigma|^{-\frac{n}{2}}\exp\big[-\frac{\sum_{i=1}^n(Y_{i}-X_{i}\beta)\Sigma^{-1}(Y_{i}-X_{i}\beta)^T}{2}\big], 
\end{align*}
along with priors 
\begin{align*}
\nu(\Sigma|\Psi, v) =&  \frac{|\Psi|^{\frac{v}{2}}}{2^{\frac{v m}{2}}\Gamma_m(\frac{v}{2})}|\Sigma|^{-\frac{\nu+m+1}{2}}\exp\big(-\frac{1}{2}tr(\Psi \Sigma^{-1})\big)\equiv W^{-1}(\Psi, v),\\
\nu(\beta \vert \lambda,\alpha) =& \left(\dfrac{\alpha}{2^{1/\alpha+1}\Gamma(1/\alpha)}\right)^{mp}
\left(\prod_{j=1}^{mp}\lambda_j\right)^{1/\alpha}
\exp\left\{-\frac{1}{2}\sum_{j=1}^{mp}\lambda_j|\beta_j|^\alpha\right\},\\
\nu(\lambda_j|\kappa_j, e_1, f_1, e_2, f_2)  =& (1-\kappa_j) \text{Gamma}(\lambda_j; e_1, f_1) + \kappa_j \text{Gamma}(\lambda_j; e_2, f_2),\\
\nu(\alpha) \sim& \text{Unif}(0.5,4),\\
\nu(\kappa_j) \sim& \text{Bern}(1/2).
\end{align*}
Consequently, the posterior full conditionals are, 
\begin{align*}
\Sigma|\beta, \Psi, v \sim&W^{-1}(\Psi +(Y-X\beta)^T(Y-X\beta), v + n)\\
q(\beta_j|\beta_{-j}, \lambda_j, \alpha) \propto&\exp\big[-\frac{tr((Y-X\beta)^T(Y-X\beta)\Sigma^{-1})+\lambda_j|\beta_j|^\alpha}{2}\big]\\
q(\alpha|\beta, \lambda) \propto& \left(\dfrac{\alpha}{2^{1/\alpha+1}\Gamma(1/\alpha)}\right)^{mp}
\left(\prod_{j=1}^{mp}\lambda_j\right)^{1/\alpha}
\exp\left\{-\frac{1}{2}\sum_{j=1}^{mp}\lambda_j|\beta_j|^\alpha\right\}\\
\pi(\lambda_i \vert  Y,  \beta, \alpha, \kappa_i,e,f) =& (1-\kappa_i)\text{Gamma}(e_1+1/\alpha, f_1+\dfrac{|\beta_i|^\alpha}{2})\\&+\kappa_i\text{Gamma}(e_2+1/\alpha, f_2+\dfrac{|\beta_i|^\alpha}{2})\\
\pi(\kappa_i|\lambda_i,e,f) \sim& \text{Bern}(\frac{\omega_2}{\omega_1+\omega_2}).
\end{align*}
We then again apply a deterministic scan component-wise MCMC algorithm. Apparently for $\Sigma$, $\lambda$, and $\kappa$ it is a direct update from known distribution while for the remaining $\beta$ and $\alpha$ we use the same Metropolis-Hastings algorithm as in the scalar response version.   
%%%%%%%%%%%%%%%%
%\end{comment}
%%%%%%%%%%%%%%%%
\end{appendix}

\bibliography{BPRref}
\bibliographystyle{apalike}
\end{document}